\documentclass[11pt]{article}
\usepackage[english]{babel}
\usepackage[utf8x]{inputenc}
\usepackage[T1]{fontenc}

\usepackage[mathlines]{lineno}

\usepackage[letterpaper, portrait, margin=1in,top=1in,bottom=1.5in]{geometry}

\usepackage{amsmath}
\usepackage{algorithm}
\usepackage{algpseudocode}
\usepackage{bbm}
\usepackage{amssymb}
\usepackage{amsthm}
\usepackage{amsfonts}
\usepackage{mathrsfs}
\usepackage{tikz}
\usepackage{graphicx}
\usepackage{textcomp}
\usepackage[shortlabels]{enumitem} 
\usepackage{graphicx,subcaption}

\newtheorem{theorem}{Theorem}[section]

\newtheorem{proposition}[theorem]{Proposition}
\newtheorem{conj}{Conjecture}[section]
\newtheorem{defn}{Definition}[section]
\theoremstyle{remark}
\newtheorem{remark}{Remark}

\newcommand{\yhat}{\hat{\mathbf{y}}}

\allowdisplaybreaks

\title{Effects of Backtracking on PageRank}
\author{Cory Glover\footnote{Network Science Institute, Northeastern University, Boston, MA \emph{glover.co@northeastern.edu}}, Tyler Jones\footnote{Amazon, \emph{tyler.j.jones@j1sfamily.com}}, Mark Kempton\footnote{Brigham Young University, Provo, UT, \emph{mkempton@mathematics.byu.edu}}, Alice Oveson\footnote{University of Maryland, College Park, MD, \emph{alice.oveson@gmail.com}}}
\date{}

\begin{document}

\maketitle
\begin{abstract}
    \noindent In this paper, we consider three variations on standard PageRank: Non-backtracking PageRank, $\mu$-PageRank, and $\infty$-PageRank, all of which alter the standard formula by adjusting the likelihood of backtracking in the algorithm's random walk. We show that in the case of regular and bipartite biregular graphs, standard PageRank and its variants are equivalent. We also compare each centrality measure and investigate their clustering capabilities. 
\end{abstract}


\section{Introduction}
\quad Since its development in 1998, the PageRank algorithm has been a powerful tool in search engine optimization \cite{page1999pagerank} and has subsequently been adapted to solve a variety of problems including word sense disambiguation in natural language processing \cite{mihalcea2004pagerank}, identifying key social media users in social network analysis \cite{heidemann2010identifying}, spam detection in web browsing \cite{gyongyi2004combating}, anomaly detection in movements of seniors \cite{payandeh2019application}, among many others. In this paper we explore various variants of PageRank: Non-backtracking PageRank as introduced by Arrigo et.~al.~\cite{arrigo2019non}, $\mu$-PageRank as introduced by Aleja et.~al.~\cite{aleja2019non} and Criado et.~al.~\cite{criado2019alpha}, and we introduce a new variant, $\infty$-PageRank. 
We also investigate the clustering and centrality measure capabilities of PageRank and its variations.

\quad Using iterated random walks, PageRank explores a network of nodes, ranking nodes by the number and quality of connections they have.
This algorithm's ability to process large networks of nodes and identify those of greatest importance or influence makes it applicable in nearly every field with only minor tailoring required to fit the different applications. One such alteration involves swapping the usual simple random walk used in the algorithm for a non-backtracking random walk. The differences between the two methods are discussed in Sections \ref{subsection:classic-pagerank} and \ref{subsection:non-backtracking-pagerank}. 

\quad Additionally, the numerical emphasis PageRank places on connections between nodes makes it well-suited for use in clustering algorithms. This potential has been thoroughly studied  \cite{andersen2006local,graham2010finding,liu2016k}. Researchers found that when compared to other clustering algorithms ($k$-means, spectral clustering, etc.). PageRank was able to deal more effectively with outliers, non-convex clusters, and high dimensional datasets \cite{liu2016k}. 
We will develop our own clustering algorithm using PageRank variants in Section \ref{section:clustering}.

\quad A non-backtracking walk on a graph $G$ is a random walk which prohibits immediate backtracking. Non-backtracking random walks are better able to model systems that are unlikely to visit the same node multiple times in a short period. 
Kitaura et al.~\cite{kitaura2022random} suggest that among these networks are models of memory and local awareness. 
Torres et al.~\cite{torres2021nonbacktracking} have also used non-backtracking to study target immunization of networks.
It is conjectured that the mixing rate (the rate at which the random walk converges to the unique stationary distribution of the graph) of a non-backtracking random walk is faster than that of a simple random walk (a conjecture proved for a variety of cases) \cite{alon2007non,kempton2016non}. 
This faster mixing rate allows for computationally efficient sampling of vertices, something that has sparked its use in many other applications including an algorithm to replace the power-of-d choices policy in allocation problems \cite{tang2018balanced}, as well as influence maximization \cite{pan2016influence} and calculation of the clustering coefficient of online social networks \cite{iwasaki2018estimating}.

\quad This study of non-backtracking random walks motivates our work on PageRank variants.
We build off the work of Arrigo et al.~\cite{arrigo2019non}, Aleja et al.~\cite{aleja2019non}, and Criado et al.~\cite{criado2019alpha} to better understand non-backtracking PageRank and $\mu$-PageRank.
In Section \ref{section:def-pagerank} we will recall the definitions of non-backtracking, and $\mu$-PageRank.
In Section \ref{section:mu-pagerank-analysis} we will define $\infty$-PageRank and analyze properties of the PageRank variants.
In Section \ref{section:clustering} we will define an algorithm for clustering networks based on PageRank and its variants.

\section{Standard, Non-Backtracking, and $\mu$-PageRank}
\label{section:def-pagerank}

\subsection{Standard PageRank}
\label{subsection:classic-pagerank}
Let $G=(V,E)$ be a directed graph with $n$ vertices and $m$ edges.
PageRank ranks the nodes of $G$ based on the number of connections a node has and the quality of those connections.
It considers a modified random walk $\{v_1,v_2,...\}$ across $G$ where,
with probability $\epsilon$, $v_i$ is chosen uniformly at random from the neighbors of $v_{i-1}$.
With probability $1-\epsilon$, $v_i$ is chosen from the set of nodes $V$ with probability distribution $u$.
The stationary distribution of the modified random walk is the PageRank vector $\pi$ of $G$, where the $i^{th}$ entry of $\pi$ is the PageRank value of node $i$.

The \emph{adjacency matrix} of the graph $G$ is the $n\times n$ matrix whose rows and columns are indexed by the vertices of $G$, and whose $ij$ entry is 1 if the directed edge $(i,j)$ occurs in the graph, and is 0 otherwise. Let $D$ denote the diagonal matrix whose diagonal entries are the out-degrees of the corresponding vertices.

\begin{defn}[PageRank]
Let $G$ be a graph with adjacency matrix $A$ and diagonal degree matrix $D$.
Let $\epsilon\in(0,1)$ and $u$ be a initial distribution vector such that $\|u\|_1=1$.
Then the stationary distribution of
\[P=\epsilon A^TD^{-1}+(1-\epsilon)u\mathbf{1}^T\]
is the PageRank vector $\pi$ of $G$.
\label{def:pagerank}
\end{defn}

Non-backtracking PageRank is the same as PageRank but the random walk step is now a non-backtracking random walk \cite{arrigo2019non,aleja2019non}.
In order to compare standard and non-backtracking PageRank, we look at an alternative representation of standard PageRank.
We will view the random walk as a walk on the directed edges of the graph.  To this end, define a graph $\hat{G}$ whose vertex set is the set of directed edges of $G$.
We put a directed edge from $(i,j)$ to $(k,l)$ in $\hat{G}$ if $j=k$.
We encapsulate this information in a matrix $C$:
\begin{align*}
    C((i,j),(k,l))=\begin{cases}1&j=k\\0&j\neq k\end{cases}.
\end{align*}
Let $\hat{D}$ be the diagonal out-degree matrix of $\hat{G}$.
Additionally we define the matrices
\begin{align*}
    S((i,j),x)&=\begin{cases}1&j=x\\0&j\neq x\end{cases}&T(x,(i,j))&=\begin{cases}1&i=x\\0&i\neq x\end{cases}
\end{align*}
which are used to project from the vertices of $G$ to its ``lifted'' edges and vice versa.
Arrigo et.~al.~\cite{arrigo2019non} use these matrices to calculate PageRank of $G$ using the graph $\hat{G}$.

When working with an undirected graph $G$, we simply view it as a directed graph in which each edge $i\sim j$ in $G$ is viewed as two directed edges $(i,j)$ and $(j,i)$.

\begin{defn}[Edge PageRank \cite{arrigo2019non}]
Let $G$ be a graph with lifted graph $\hat{G}$.
Let $C$ be the adjacency matrix of $\hat{G}$ and $\hat{D}$ the (edge) degree matrix.
Let $u=T^T(TT^T)^{-1}v$ be the initial distribution vector from standard PageRank projected onto $\hat{G}$.
Then the stationary distribution of 
\[H_1=\epsilon C^T\hat{D}^{-1}+(1-\epsilon)u\mathbf{1}^T\] 
is the edge PageRank $\hat{\pi}$ of $G$.
\label{def:edge-pagerank}
\end{defn}

\begin{theorem}[\cite{arrigo2019non} Corollary 1]
In Definition 2.2, The PageRank of $G$ is $\pi=T\hat{\pi}$.
\label{def:edge-pagerank}
\end{theorem}

\subsection{Non-backtracking PageRank}
\label{subsection:non-backtracking-pagerank}

Using the lift of $G$ onto its directed edges, we can calculate non-backtracking PageRank.
To do this we consider the graph created by the non-backtracking matrix defined by
\[B((i,j),(k,l))=\begin{cases}1&j=k\text{ and }i\neq l\\0&\text{otherwise}\end{cases}.\]
The graph associated with $B$ has much in common with $\hat{G}$, however edges that backtrack on $G$ do not connect to each other.
Note that performing a random walk on the graph generated by $B$ is equivalent to a non-backtracking random walk on the graph $G$.
Hence, using $B$ as our adjacency matrix, we can define non-backtracking PageRank in the same manner as Arrigo et~al.~\cite{arrigo2019non} as follows:

\begin{defn}[Non-backtracking PageRank \cite{arrigo2019non}]
Let $G$ be a graph with $B$ its non-backtracking matrix and $\hat{D}$ its edge degree matrix.
Let $u=T^T(TT^T)^{-1}v$ be the initial distribution vector from standard PageRank projected onto the directed edges of $G$.
Then the stationary distribution of 
\[H_0=\epsilon B^T(\hat{D}-I)^{-1}+(1-\epsilon)u\mathbf{1}^T\]
is the edge non-backtracking PageRank $\hat{\pi}_0$ of $G$.
The non-backtracking PageRank of $G$ is $\pi_0=T\hat{\pi}_0$.
\end{defn}

\subsection{$\mu$-PageRank}
\label{subsection:mu-pagerank}

The notion of $\mu$-PageRank (or $\mu$-centrality) has been studied to give an alternative ranking to nodes \cite{aleja2019non,criado2019alpha}.
The main idea of $\mu$-PageRank is to again perform a modified random walk $\{v_1,v_2,...\}$ where $v_i$ is chosen at random with probability $\epsilon$ from the neighbors of $v_{i-1}$.
However in $\mu$-PageRank, the neighbor from $v_{i-2}$ is weighted by a factor of $\mu$ and all other edges are equally likely to be chosen.
To encapsulate this modified random walk into a Markov chain one can lift the graph $G$ to a graph of directed edges.
We weight every backtracking connection with the probability $\mu$.
To do this we define a backtracking operator $\tau$:
\[\tau((i,j),(k,l))=\begin{cases}1&j=k,i=l\\0&\text{otherwise}\end{cases}.\]

\begin{defn}[$\mu$-PageRank]
Let $G$ be a graph with $C$ its edge adjacency matrix, $\hat{D}$ its edge degree matrix and $\tau$ the backtracking operator.
Let $C_\mu=C-(1-\mu)\tau$.
Let $u=T^T(TT^T)^{-1}v$ be the initial distribution vector from standard PageRank projected onto the directed edges of $G$.
Then the stationary distribution of
 \[H_\mu=\epsilon C_\mu^T(\hat{D}-(1-\mu)I)^{-1}+(1-\epsilon)u\mathbf{1}^T\]
is the edge $\mu$-PageRank $\hat{\pi}_\mu$ of $\hat{G}$.
The $\mu$-PageRank of $G$ is $\pi_\mu=T\hat{\pi}_\mu$.
\label{defn:mu-pagerank}
\end{defn}

\begin{remark}
Aleja et al.~and Criado et al.~\cite{aleja2019non,criado2019alpha} also study $\mu$-PageRank for values $\mu\in[0,1]$.
Our study allows $\mu\in[0,\infty)$ to represent random walks where backtracking becomes increasingly likely.
This notion will be studied in section \ref{subsection:limit}.
\end{remark}

\begin{remark}
Setting $\mu=1$ weights all the edges the same giving the standard PageRank of $G$.
If $\mu=0$, this makes it impossible to choose $v_{i-2}$ as $v_i$ with probability $\epsilon$ and is therefore non-backtracking PageRank.
This motivates the notation $H_1$ and $H_0$ for standard and non-backtracking PageRank respectively.
\label{remark:mu-equals-1}
\end{remark}

\section{Analysis of $\mu$-PageRank}
\label{section:mu-pagerank-analysis}

\quad Given the versatility of $\mu$-PageRank, analyzing $\mu$-PageRank can lead to insights for both non-backtracking and standard PageRank.
In this section, we will prove general properties of $\mu$-PageRank.
We will also calculate the limiting value of $\mu$-PageRank as $\mu\rightarrow\infty$, defining the notion of $\infty$-PageRank.
We will end the section discussing conjectures about the implications of $\infty$-PageRank.

\subsection{$\mu$-PageRank of Regular Graphs and Biregular Bipartite Graphs}
\label{subsec:reg-bireg}
\quad With the same methodology as Arrigo et al.~\cite{arrigo2019non}, we can prove a more general result of their Theorem 4, which shows that when the graph is regular (all nodes with the same out-degree), then the nonbacktracking PageRank is the same as the usual PageRank.  We show this holds for the $\mu$-PageRank for any $\mu\geq0$.

\begin{theorem}
Let $A$ be the adjacency matrix of a directed $k$-regular connected graph with $k\geq 2$.
Then for $\epsilon\in[0,1]$, define the matrices $H_\mu$, $H_1$ and $T$ as before.
Then for any $\mu\in[0,\infty)$, if $H_\mu^T\hat{\mathbf{y}}=\hat{\mathbf{y}}$ and $H_1^T\hat{\mathbf{x}}=\hat{\mathbf{x}}$ with $\|\hat{\mathbf{x}}\|=\|\hat{\mathbf{y}}\|=1$, then we have $T\hat{\mathbf{y}}=T\hat{\mathbf{x}}$.  That is, $\pi_\mu=\pi$ for all $\mu$.
\end{theorem}

\begin{proof}
We follow the proof of Theorem 4 in \cite{arrigo2019non} with $H_\mu$ in place of $H_0$.  We obtain $\hat{\mathbf{y}}$ by solving
\[
\left(I-\frac{\epsilon}{k-(1-\mu)}C_\mu^T\right)\yhat = \frac{1-\epsilon}{nk}T^T\mathbf{1} = \frac{1-\epsilon}{nk}\mathbf{1}
\] hence 
\[
\yhat = \frac{1-\epsilon}{nk}\left(I-\frac{\epsilon}{k-(1-\mu)}C_\mu^T\right)^{-1}\mathbf{1}.
\]

Now, we have $C_\mu^T\mathbf{1} = (k-(1-\mu))\mathbf{1}$, so for $\epsilon<1$ we have that all the eigenvalues of $\frac{\epsilon}{k-(1-\mu)}C_\mu^T$ are less than 1, so we may express $\yhat$ as a geometric series:
\begin{align*}
\yhat&=\frac{1-\epsilon}{nk}\left(\sum_{r=0}^\infty\frac{\epsilon^r}{(k-(1-\mu))^r}(C_\mu^T)^r\right)\mathbf{1}\\
&=\frac{1-\epsilon}{nk}\sum_{r=0}^\infty\frac{\epsilon^r}{(k-(1-\mu))^r}(k-(1-\mu))^r\mathbf{1}\\
&=\frac{1-\epsilon}{nk}\frac{1}{1-\epsilon}\mathbf{1}\\
&=\frac{1}{nk}\mathbf{1}.
\end{align*}

Hence we get that the projection satisfies
\[
T\yhat = \frac{1}{nk}T\mathbf{1} = \frac{1}{n}\mathbf{1}=\pi
\]
as desired.
\end{proof}

We can further extend the result to bipartite biregular graphs, which are bipartite graphs where every node within any one of the parts has the same out-degree.

\begin{theorem}
Let $A$ be the adjacency matrix of a directed bipartite biregular connected graph with $d_1,d_2\geq 2$ where $d_i$ is the degree of each node in the $i^{th}$ part. Then for $\epsilon\in(0,1)$, define the matrices $H_\mu$, $H_1$ and $T$ as before.
If $H_1^T\hat{\mathbf{x}}=\hat{\mathbf{x}}$ and $H_\mu^T\hat{\mathbf{y}}=\hat{\mathbf{y}}$, and $\|\hat{\mathbf{x}}\|=\|\hat{\mathbf{y}}\|=1$, then $T\hat{\mathbf{y}}=T\hat{\mathbf{x}}$.  That is, $\pi_\mu=\pi$ for all $\mu>0$.
\label{theorem:bipartite}
\end{theorem}

\begin{proof}
Since $G$ is bipartite, we can write
\begin{align*}
    C_\mu=\begin{pmatrix}0&C_2\\C_1&0\end{pmatrix}&&T=\begin{pmatrix}T_1&0\\0&T_2\end{pmatrix}
\end{align*}
where the $\mu$ denotes the backtracking parameter and the subscripts 1 and 2 denote the first and second parts of the bipartite graph.
We organize the columns of the matrices such that the first $r$ columns are associated with the $r$ nodes in the first part and the $m-r$ columns are associated with the $m-r$ nodes in the second part.
We can solve for $\hat{\mathbf{y}}$ by solving
\[\Biggl(\begin{pmatrix}I_1&0\\0&I_2\end{pmatrix}-\epsilon\begin{pmatrix}0&C_1^T\\C_2^T&0\end{pmatrix}\begin{pmatrix}\frac{1}{d_2-(1-\mu)}I&0\\0&\frac{1}{d_1-(1-\mu)}I\end{pmatrix}\Biggr)\hat{\mathbf{y}}=\frac{1-\epsilon}{n}\begin{pmatrix}T_1^T&0\\0&T_2^T\end{pmatrix}\begin{pmatrix}\frac{1}{d_1}&0\\0&\frac{1}{d_2}\end{pmatrix}\mathbf{1}.\]

Recall that $T$ is the out degree matrix. Thus $T^T\mathbf{1}$ counts the number of nodes which point to a certain edge (the edge in-degree). Note that edges can only have one node pointing to them. So $T^T\mathbf{1}=\mathbf{1}$. We can then simplify the right side of the equation to
\begin{align*}
    \frac{1-\epsilon}{n}\begin{pmatrix}T_1^T&0\\0&T_2^T\end{pmatrix}\begin{pmatrix}\frac{1}{d_1}&0\\0&\frac{1}{d_2}\end{pmatrix}\mathbf{1}=\frac{1-\epsilon}{n}\begin{pmatrix}\frac{1}{d_1}\mathbf{1}\\\frac{1}{d_2}\mathbf{1}\end{pmatrix}.
\end{align*}

We use the inverse of the matrix on the left side of the equation to get
$$\hat{\mathbf{y}}=\frac{1-\epsilon}{n}\Biggl(\begin{pmatrix}I_1&0\\0&I_2\end{pmatrix}-\epsilon\begin{pmatrix}0&\frac{1}{d_1-(1-\mu)}C_1^T\\\frac{1}{d_2-(1-\mu)}C_2^T&0\end{pmatrix}\Biggr)^{-1}\begin{pmatrix}\frac{1}{d_1}\mathbf{1}\\\frac{1}{d_2}\mathbf{1}\end{pmatrix}.$$
We now want to replace the inverse matrix with a geometric series.
To do so, we verify that the eigenvalues of $\epsilon\begin{pmatrix}0&\frac{1}{d_1-(1-\mu)}C_1^T\\\frac{1}{d_2-(1-\mu)}C_2^T&0\end{pmatrix}$ have absolute value less than 1.
We note that $C_1^T\mathbf{1}$ will count the number of incoming edges to partition 1 from partition 2. This is the same as counting the number of edges leaving a node from partition 1 to partition 2 (the degree of nodes in partition 1). Thus, $C_1^T\mathbf{1}=(d_1-(1-\mu))\mathbf{1}$. Similarly $C_2^T\mathbf{1}=(d_2-(1-\mu))\mathbf{1}$.
Hence we find that
\begin{align*}
    \epsilon\begin{pmatrix}0&\frac{1}{d_1-(1-\mu)}C_1^T\\\frac{1}{d_2-(1-\mu)}C_2^T&0\end{pmatrix}\mathbf{1}&=\epsilon\begin{pmatrix}\frac{d_1-(1-\mu)}{d_1-(1-\mu)}\\\frac{d_2-(1-\mu)}{d_2-(1-\mu)}\end{pmatrix}=\epsilon\mathbf{1}.
\end{align*}
Thus if $\epsilon < 1$, we have that all eigenvalues are less than 1 and we can rewrite the matrix with a geometric series.
We replace the inverse matrix with the geometric series.
\begin{align*}
    \hat{\mathbf{y}}&=\frac{1-\epsilon}{n}\left(\sum_{r=0}^\infty\epsilon^r\begin{pmatrix}0&\frac{1}{d_1-(1-\mu)}C_1^T\\\frac{1}{d_2-(1-\mu)}C_2^T&0\end{pmatrix}^r\right)\begin{pmatrix}\frac{1}{d_1}\mathbf{1}\\\frac{1}{d_2}\mathbf{1}\end{pmatrix}.
\end{align*}
By induction, we can rewrite this sum based on the even and odd values of $r$.
When we do this, we have the matrices
\begin{align*}
    J_{1_r}&=\begin{pmatrix}\Bigl(\frac{1}{(d_1-(1-\mu))(d_2-(1-\mu))}C_1^TC_2^T\Bigr)^r&0\\0&\Bigl(\frac{1}{(d_1-(1-\mu))(d_2-(1-\mu))}C_2^TC_1^T\Bigr)^r\end{pmatrix}\\
    J_{2_r}&=\begin{pmatrix}0&\frac{1}{d_1-(1-\mu)}C_1^T\Bigl(\frac{1}{(d_1-(1-\mu))(d_2-(1-\mu))}C_2^TC_1^T\Bigr)^r\\\frac{1}{d_2-(1-\mu)}C_2^T\Bigr(\frac{1}{(d_1-(1-\mu))(d_2-(1-\mu))}C_1^TC_2^T\Bigr)^r&0\end{pmatrix}
\end{align*}
corresponding with the even and odd terms respectively.
This gives the PageRank equation
\begin{align*}
    \hat{\mathbf{y}}=&\frac{1-\epsilon}{n}\Biggl(\sum_{r=0}^\infty\epsilon^{2r}J_{1_r}\begin{pmatrix}\frac{1}{d_1}\mathbf{1}\\\frac{1}{d_2}\mathbf{1}\end{pmatrix}+\sum_{r=0}^\infty\epsilon^{2r+1}J_{2_r}\begin{pmatrix}\frac{1}{d_1}\mathbf{1}\\\frac{1}{d_2}\mathbf{1}\end{pmatrix}\Biggr)\\
    =&\frac{1-\epsilon}{n}\Biggl(\sum_{r=0}^\infty\epsilon^{2r}\begin{pmatrix}\frac{1}{d_1}\Bigl(\frac{(d_1-(1-\mu))(d_2-(1-\mu))}{(d_1-(1-\mu))(d_2-(1-\mu))}\mathbf{1}\Bigr)^r\\\frac{1}{d_2}\Bigl(\frac{(d_1-(1-\mu))(d_2-(1-\mu))}{(d_1-(1-\mu))(d_2-(1-\mu))}\mathbf{1}\Bigr)^r\end{pmatrix}\\&+\sum_{r=0}^\infty\epsilon^{2r+1}\begin{pmatrix}\frac{1}{d_2(d_1-(1-\mu))}C_1^T\Bigl(\frac{(d_1-(1-\mu))(d_2-(1-\mu))}{(d_1-(1-\mu))(d_2-(1-\mu))}\mathbf{1}\Bigr)^r\\\frac{1}{d_1(d_2-(1-\mu))}C_2^T\Bigr(\frac{(d_1-(1-\mu))(d_2-(1-\mu))}{(d_1-(1-\mu))(d_2-(1-\mu))}\mathbf{1}\Bigr)^r\end{pmatrix}\Biggr)\\
    =&\frac{1-\epsilon}{n}\Biggl(\sum_{r=0}^\infty\epsilon^{2r}\begin{pmatrix}\frac{1}{d_1}\mathbf{1}\\\frac{1}{d_2}\mathbf{1}\end{pmatrix}+\sum_{r=0}^\infty\epsilon^{2r+1}\begin{pmatrix}\frac{1}{d_2}\mathbf{1}\\\frac{1}{d_1}\mathbf{1}\end{pmatrix}\Biggr)\\
    =&\frac{1-\epsilon}{n}\Biggl(\begin{pmatrix}\Bigl(\frac{1}{d_1}+\frac{\epsilon}{d_2}\Bigr)\Bigl(\sum_{r=0}^\infty(\epsilon^2)^r\Bigr)\mathbf{1}\\\Bigl(\frac{1}{d_2}+\frac{\epsilon}{d_1}\Bigr)\Bigl(\sum_{r=0}^\infty(\epsilon^2)^r\Bigr)\mathbf{1}\end{pmatrix}\Biggr).
\end{align*}
Since $\epsilon^2<1$, the summation converges to $\frac{1}{1-\epsilon^2}$. Further, we simplify the fractions and get
\begin{align*}
    \hat{\mathbf{y}}&=\frac{1}{nd_1d_2(1+\epsilon)}\begin{pmatrix}(d_2+\epsilon d_1)\mathbf{1}\\(d_1+\epsilon d_2)\mathbf{1}\end{pmatrix}.
\end{align*}
Now to project $\hat{\mathbf{y}}$ to the vertex space, we get
\begin{align*}
    T\hat{\mathbf{y}}&=\frac{1}{nd_1d_2(1+\epsilon)}\begin{pmatrix}(d_2+\epsilon d_1)(T_1\mathbf{1})\\(d_1+\epsilon d_2)(T_2\mathbf{1})\end{pmatrix}.
\end{align*}

Note that $T\mathbf{1}$ counts the number of outgoing edges from a given node. Thus, $T_1\mathbf{1}=d_1$ and $T_2\mathbf{1}=d_2$.
Thus,
\begin{align*}
    T\hat{\mathbf{y}}&=\frac{1}{nd_1d_2(1+\epsilon)}\begin{pmatrix}(d_2+\epsilon d_1)d_1\mathbf{1}\\(d_1+\epsilon d_2)d_2\mathbf{1}\end{pmatrix}\\
    &=\frac{1}{n(1+\epsilon)}\begin{pmatrix}1+\epsilon\frac{d_1}{d_2}&0\\0&1+\epsilon\frac{d_2}{d_1}\end{pmatrix}\mathbf{1}.
\end{align*}

Recall that in a bipartite, biregular graph, the adjacency matrix is $A=\begin{pmatrix}0&A_2\\A_1&0\end{pmatrix}.$
Thus 
\begin{align*}
    \mathbf{x}&=\frac{1-\epsilon}{n}(I-\alpha A^TD^{-1})^{-1}\mathbf{1}\\
    &=\frac{1-\epsilon}{n}\Bigl(\sum_{r=0}^\infty\epsilon^r(A^T)^r(D^{-1})^r\Bigr)\mathbf{1}\\
    &=\frac{1-\epsilon}{n}\Bigl(\sum_{r=0}^\infty\epsilon^r\begin{pmatrix}0&A_1^T\\A_2^T&0\end{pmatrix}^r\begin{pmatrix}\frac{1}{d_2}&0\\0&\frac{1}{d_1}\end{pmatrix}^r\Bigr)\mathbf{1}\\
    &=\frac{1-\epsilon}{n}\Bigl(\sum_{r=0}^\infty\epsilon^r\begin{pmatrix}0&\frac{A_1^T}{d_2}\\\frac{A_2^T}{d_1}&0\end{pmatrix}^r\Bigr)\mathbf{1}\\
    &=\frac{1-\epsilon}{n}\Bigl(\sum_{r=0}^\infty\epsilon^{2r}\begin{pmatrix}(\frac{A_2^TA_1^T}{d_1d_2})^r&0\\0&(\frac{A_1^TA_2^T}{d_1d_2})^r\end{pmatrix}+\sum_{r=0}^\infty\epsilon^{2r+1}\begin{pmatrix}0&\frac{A_1^T}{d_2}(\frac{A_1^TA_2^T}{d_1d_2})^r\\\frac{A_2^T}{d_1}(\frac{A_2^TA_1^T}{d_1d_2})^r&0\end{pmatrix}\Bigr)\mathbf{1}\\
    &=\frac{1-\epsilon}{n}\Bigl(\sum_{r=0}^\infty\epsilon^{2r}\begin{pmatrix}(\frac{A_2^TA_1^T}{d_1d_2})^r\mathbf{1}\\(\frac{A_1^TA_2^T}{d_2d_1})^r\mathbf{1}\end{pmatrix}+\sum_{r=0}^\infty\epsilon^{2r+1}\begin{pmatrix}\frac{A_1^T}{d_2}(\frac{A_1^TA_2^T}{d_2d_1})^r\mathbf{1}\\\frac{A_2^T}{d_1}(\frac{A_2^TA_1^T}{d_1d_2})^r\mathbf{1}\end{pmatrix}\Bigr)\\
    &=\frac{1-\epsilon}{n}\Bigl(\sum_{r=0}^\infty\epsilon^{2r}\mathbf{1}+\sum_{r=0}^\infty\epsilon^{2r+1}\begin{pmatrix}\frac{d_1}{d_2}\\\frac{d_2}{d_1}\end{pmatrix}\mathbf{1}\Bigr)\\
    &=\frac{1-\epsilon}{n}\Bigl(\sum_{r=0}^\infty\epsilon^{2r}\mathbf{1}+\epsilon\sum_{r=0}^\infty\epsilon^{2r}\begin{pmatrix}\frac{d_1}{d_2}\\\frac{d_2}{d_1}\end{pmatrix}\mathbf{1}\Bigr)\\
    &=\frac{1-\epsilon}{n(1-\epsilon^2)}\Bigl(\begin{pmatrix}1+\epsilon\frac{d_1}{d_2}\\1+\epsilon\frac{d_2}{d_1}\end{pmatrix}\mathbf{1}\Bigr)\\
    &=\frac{1}{n(1+\epsilon)}\begin{pmatrix}1+\epsilon\frac{d_1}{d_2}&0\\0&1+\epsilon\frac{d_2}{d_1}\end{pmatrix}\mathbf{1}.
\end{align*}
Thus, $T\hat{\mathbf{y}}=\mathbf{x}=T\hat{\mathbf{x}}$.
\end{proof}

\begin{remark}[\cite{arrigo2019non}, corollaries 1,2]
Applying Theorem \ref{theorem:bipartite} with $\mu=1$ and $\mu=0$, we have the bipartite biregular graphs result in equivalent PageRank values for standard and non-backtracking PageRank.
\end{remark}

\subsection{Limit as $\mu\rightarrow\infty$}
\label{subsection:limit}

We remark that in everything above, we are working with directed graphs as this provides greater generality and the definitions and results apply to this level of generality without any added difficulty.  However, the non-backtracking condition is less meaningful in directed graphs when the reversal of some edges may not be present. As such, we will focus the remainder of the paper on undirected graphs, which can be viewed as directed graphs by considering each undirected edge as a pair of directed edges that are the reversals of each other.  The proof of the next result will rely on the graph being undirected.

\quad Most commonly $\mu$-PageRank is studied for $\mu\in[0,1]$ (see \cite{aleja2019non,criado2019alpha}).
However we can easily extend the domain of $\mu$ to be $[0,\infty)$.
Intuitively, as $\mu$ gets larger it becomes more and more likely to backtrack with probability $\epsilon$ in the modified random walk.
Thus as $\mu\rightarrow\infty$, we approach the $\infty$-PageRank of the graph $G$ (or the ranking of nodes if with probability $\epsilon$ we bounce back and forth between two nodes).
This limit can be calculated explicitly.

\begin{figure}
    \centering
    \begin{subfigure}{.45\linewidth}
    \centering\includegraphics[width=\linewidth]{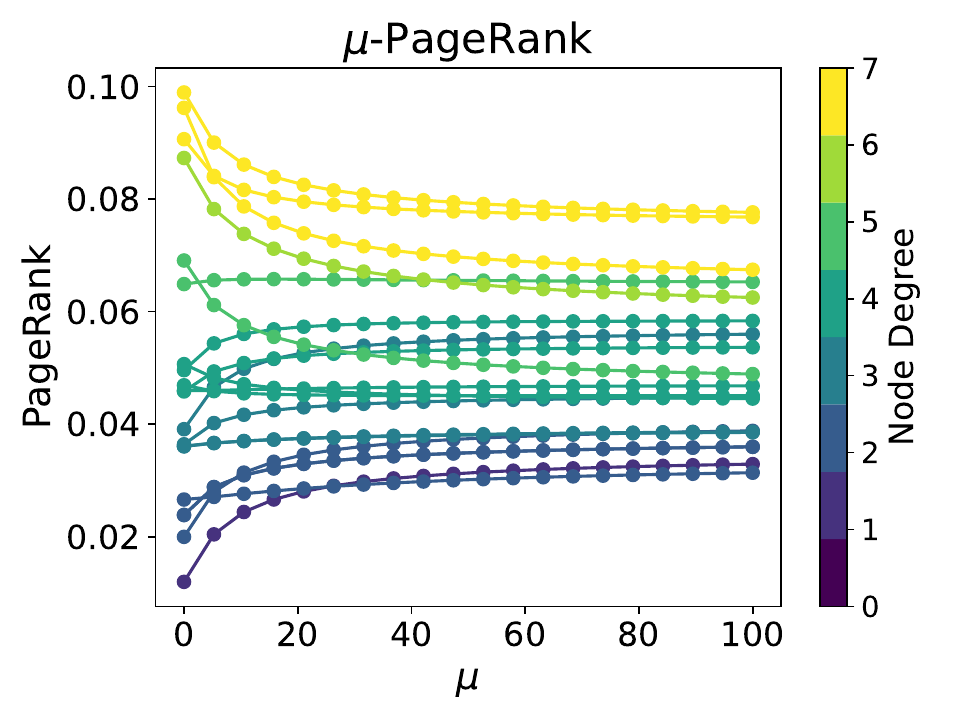}
    \subcaption{The PageRank values of an ER graph $(N=20,\langle k\rangle=4)$ over 20 different values of $\mu$ spanning from 0 to 100. Curves are colored by degree of node. }
    \end{subfigure}\hfill
    \begin{subfigure}{.45\linewidth}
    \centering\includegraphics[width=\linewidth]{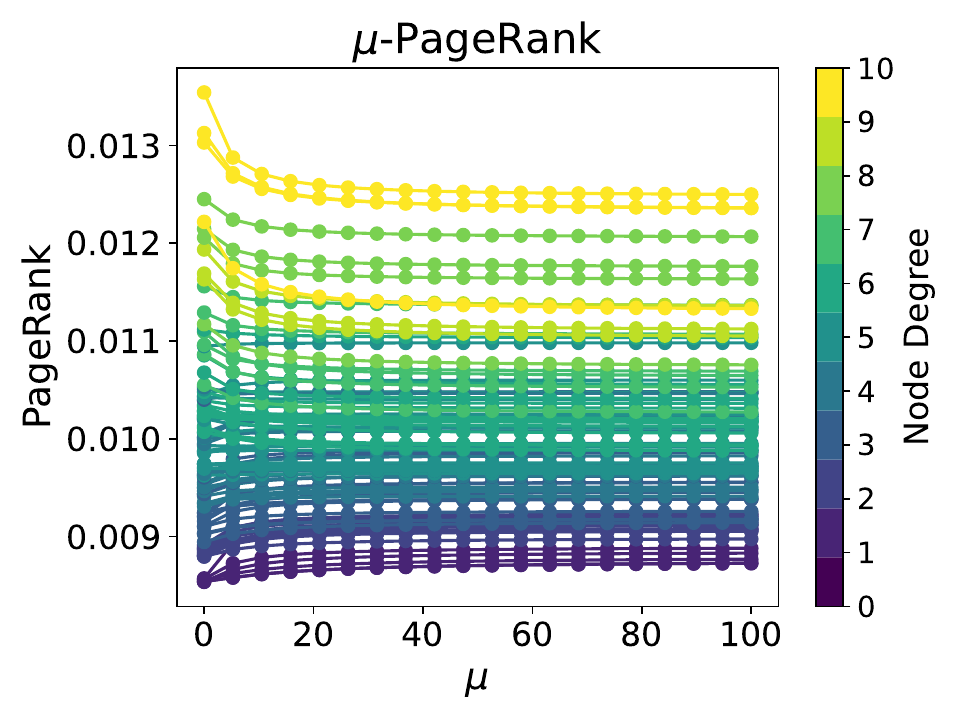}
    \subcaption{The PageRank values of an ER graph $(N=100,\langle k\rangle=4)$ over 20 different values of $\mu$ spanning from 0 to 100. Curves are colored by degree of node.}
    \end{subfigure}
    \caption{$\mu$-PageRank Values over $\mu\in[0,100]$}
    \label{fig:centrality}%
\end{figure}

\begin{proposition}[$\infty$-PageRank]
Let $G$ be an undirected graph, $v$ an initial distribution vector on the vertices of $G$, and $\mu$-PageRank as in Definition \ref{defn:mu-pagerank}.
Let $\hat{\pi}_\infty=\lim_{\mu\rightarrow\infty}\hat{\pi}_\mu$ and let $\pi_\infty = T\hat{\pi}_\infty$ be the $\infty$-PageRank of $G$.
Then 
\begin{align*}
    \pi_\infty&=(1+\epsilon)^{-1}v+\epsilon(1+\epsilon)^{-1}AD^{-1}v.
\end{align*}
\label{prop:limit}
\end{proposition}

\begin{proof}
For $\epsilon\in(0,1)$ we have
\begin{align*}
    H_\infty&=\lim_{\mu\rightarrow\infty}\epsilon C_\mu(\hat{D}-(1-\mu)I)^{-1}+(1-\epsilon)u\mathbf{1}^T\\
    &=\lim_{\mu\rightarrow\infty}\epsilon(C-(1-\mu)\tau)(\hat{D}-(1-\mu)I)^{-1}+(1-\epsilon)u\mathbf{1}^T\\
    &=\epsilon\tau+(1-\epsilon)u\mathbf{1}^T.
\end{align*}
Then we have
\begin{align}\label{eq:piinf}
    \hat{\pi}_\infty&=\epsilon\tau\hat{\pi}_\infty+(1-\epsilon)u\mathbf{1}^T\hat{\pi}_\infty.
    \end{align}
    
Observe that each $\pi_\mu$ is a distribution vector, so its entries sum to $1$ for each $\mu$, thus this will remain true in the limit, so $\mathbf{1}^T\hat{\pi}_\infty=1$.  Recall also that $u=T^T(TT^T)^{-1}v$ (see Definition \ref{defn:mu-pagerank}). Replacing this in (\ref{eq:piinf}) and simplifying, we get     
    
    \begin{align*}
    (I-\epsilon\tau)\hat{\pi}_\infty&=(1-\epsilon)T^T(TT^T)^{-1}v
\end{align*}

Since $G$ is an undirected graph, we have that $\tau^2=I$ (see \cite{kempton2016non}), so direct computation shows that $I-\epsilon\tau$ is invertible with $(I-\epsilon\tau)^{-1}=(1-\epsilon^2)^{-1}I+\frac{\epsilon}{1-\epsilon^2}\tau$. So

\begin{align*}    
  \hat{\pi}_\infty&=(1-\epsilon)(I-\epsilon\tau)^{-1}T^T(TT^T)^{-1}v\\
&=\Big((1+\epsilon)^{-1}I+\epsilon(1+\epsilon)^{-1}\tau\Big)T^T(TT^T)^{-1}v.
\end{align*}
To project onto the nodes of $G$, we left multiply by $T$.
To simplify this computation, we use the facts that $TT^T=D$, $T\tau=S^T$ (since $G$ is undirected), and $TS=A$ (see \cite{kempton2016non}).
This gives
\begin{align*}
    \pi_\infty&=(1+\epsilon)^{-1}TT^T(TT^T)^{-1}v+\epsilon(1+\epsilon)^{-1}T\tau T^T(TT^T)^{-1}v\\
    &=(1+\epsilon)^{-1}v+\epsilon(1+\epsilon)^{-1}AD^{-1}v.
\end{align*}
\end{proof}

Unsurprisingly, the $\infty$-PageRank of a node only depends upon the degree of its neighbors.
It is worth noting that this limit can be calculated extremely quickly.

\begin{remark}
The proof of Proposition \ref{prop:limit} relied on $G$ being an undirected graph: when we used $\tau^2=I$ to compute the inverse of $I-\epsilon\tau$ and the fact that $T\tau=S^T$.  These facts fail when $G$ is a directed graph in which some directed edges do not have a reversal as an edge. Indeed, diagonal entries of $\tau^2$ will be 0 for edges whose reversal is missing in the graph.  The inverse of $I-\epsilon\tau$ could still be computed, but its exact form depends on which edges have reversals and which do not.  Likewise, this will effect the exact form of $T\tau T^T$.  For any given directed case, the limit in the proof could still be worked out and and expression for $\pi_\infty$ obtained, but will depend on exactly which edges have reversals and which do not.
\end{remark}

\subsection{Monotonicity}
\label{subsection:monotonicity}

\quad As $\mu\rightarrow\infty$, we the modified random walk will simply travel back and forth between two nodes with probability $\epsilon$ and randomly jump to another set of nodes with probability $1-\epsilon$.
Consequently, for large $\mu$ we do not expect the value of $\mu$-PageRank for a single node.
Furthermore, for nodes of large degree we would expect there $\mu$-PageRank value to decrease as $\mu\rightarrow\infty$ as the chances of the random walk arriving at said node will decrease due to the modified random walk getting stuck at lower degree nodes.
We also expect the $\mu$-PageRank value of low degree nodes to increase, as a low degree node and it's neighbors have the same initial probability of being chosen in a random jump as a high degree node and the resultant modified random walk will have many more steps to this low degree node.

To test this intuition, we generated a two random graphs with average degree 4, one of size 20 and one of size 100, and calculated the $\mu$-PageRank centrality values of the graph for 20 different values of $\mu$ between 0 and 100. This simulates taking the limit as $\mu$ goes to infinity as we have seen most graphs of this nature have $\mu$-PageRank empirically converge to their limit before 100. This produces the plots seen in Figure~\ref{fig:centrality}. 
As expected, we see that the $\mu$-PageRank converges as $\mu\rightarrow\infty$ and that high degree nodes decrease in their value and low degree nodes increase in their value.
Additionally, we empirically see that this convergence is monotonic.
Based on these observations, we make the following conjecture:
\begin{conj}
The function $(\pi_\mu)_i$ is monotonic for all $i$ and \[\max_i\pi_\mu(i)-\min_i\pi_\mu(i)\leq\max_i\pi_{\mu+\alpha}(i)-\min_i\pi_{\mu+\alpha}(i)\] for all $\alpha>0$.
\end{conj}

If this conjecture holds, it leads to computationally efficient comparisons between standard and non-backtracking PageRank.
Specifically, one can compare $\infty$-PageRank with standard PageRank to determine whether the non-backtracking PageRank of a node is greater or lesser than standard PageRank. 
Due to the computational complexity of calculating non-backtracking PageRank, this can create a quick comparison between the node standard and non-backtracking PageRank values and rankings.


\subsection{Top Nodes in $\infty$-PageRank and Standard PageRank}
\label{subsec:top-nodes}

\quad Due to precision error, the exact PageRank value of a node is often not as informative as the nodes ranking in comparison to other vertices in the network.
As such, we investigate how the top nodes of standard PageRank compare with the top nodes of $\infty$-PageRank.
This is motivated by the computational efficiency of computing $\infty$-PageRank.

Many real-world networks exhibit scale-free behavior \cite{barabasi2009scale,barabasi2003scale}.
That is, the degrees of the nodes are drawn from a heavy-tailed distribution.
The standard PageRank of scale-free networks has been previously studied \cite{avrachenkov2006pagerank,garavaglia2020local}.
In fact, \cite{garavaglia2020local} finds if the degree distribution is heavy-tailed, the PageRank will also be heavy-tailed.
Given the similarity in dependence on the degree distribution of $\infty$-PageRank, we expect similar results (see Figure \ref{fig:dist}a).

\begin{conj}
If $d_i\sim X$ where $X$ is a regularly-varying random variable and $d_i$ is the degree of node $i\in V(G)$, then $\pi_\infty(i)$ will also be scale-free.
\end{conj}

We further hypothesize these distributions will mirror each other.
To test this, we measure the overlap of top nodes in each PageRank distribution.
We test this on random graphs and on scale-free networks generated by a hyper-soft configuration model based on a Pareto distribution \cite{voitalov2019scale}.
As seen in Figure \ref{fig:dist}b, we find high overlap for scale-free networks consistently for the tops nodes.
This result does not hold for random graphs.
However, as we consider a larger set of nodes, we see that the overlap decreases for scale-free networks.
This is likely due to the ``importance'' of semi-important nodes in standard PageRank being shifted to low-degree nodes connected to hubs in $\infty$-PageRank.

Because many real systems exhibit large degree heterogeneity, we examine the overlap between standard and $\infty$-PageRank and nine real systems \cite{peixoto2020netzschleuder}. We see that the overlap follows a similar pattern to that of network models with large degree heterogeneity (see Fig.~\ref{fig:real-networks}a).
Again, this is likely due to shifting importance to low-degree nodes connected to hubs in $\infty$-PageRank as seen in the foodweb shown in Fig.~\ref{fig:real-networks}b,c.

\begin{figure}
    \centering
    \begin{subfigure}{.45\linewidth}\centering
        \includegraphics[width=\linewidth]{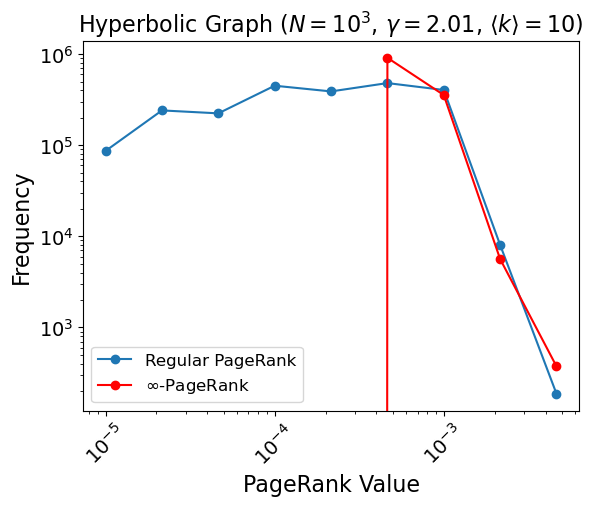}
        \caption{The PR and $\infty$-PR distributions for a hyperbolic random graph with a degree exponent of $\gamma=2.01$.}
    \end{subfigure}
    \begin{subfigure}{.45\linewidth}\centering
        \includegraphics[width=\linewidth]{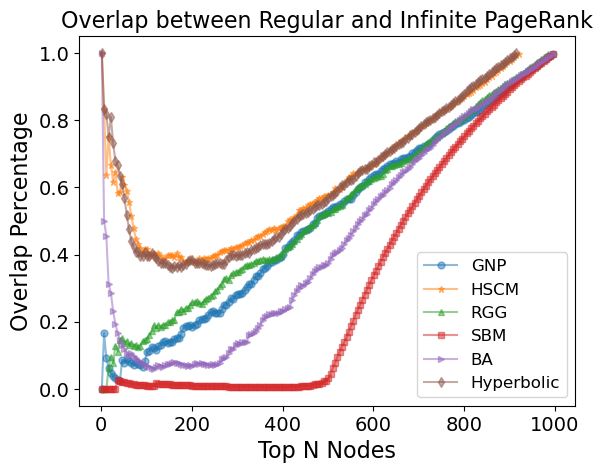}
        \caption{Percentage of overlap of top $N$ nodes between standard and $\infty$-PageRank for multiple graph models.}
    \end{subfigure}
    \caption{}
    \label{fig:dist}
\end{figure}

\begin{figure}
\sbox0{\begin{subfigure}[b]{\dimexpr 0.7\textwidth-0.3\columnsep}
  \centering\hspace*{-1.5cm}\includegraphics[height=3in]{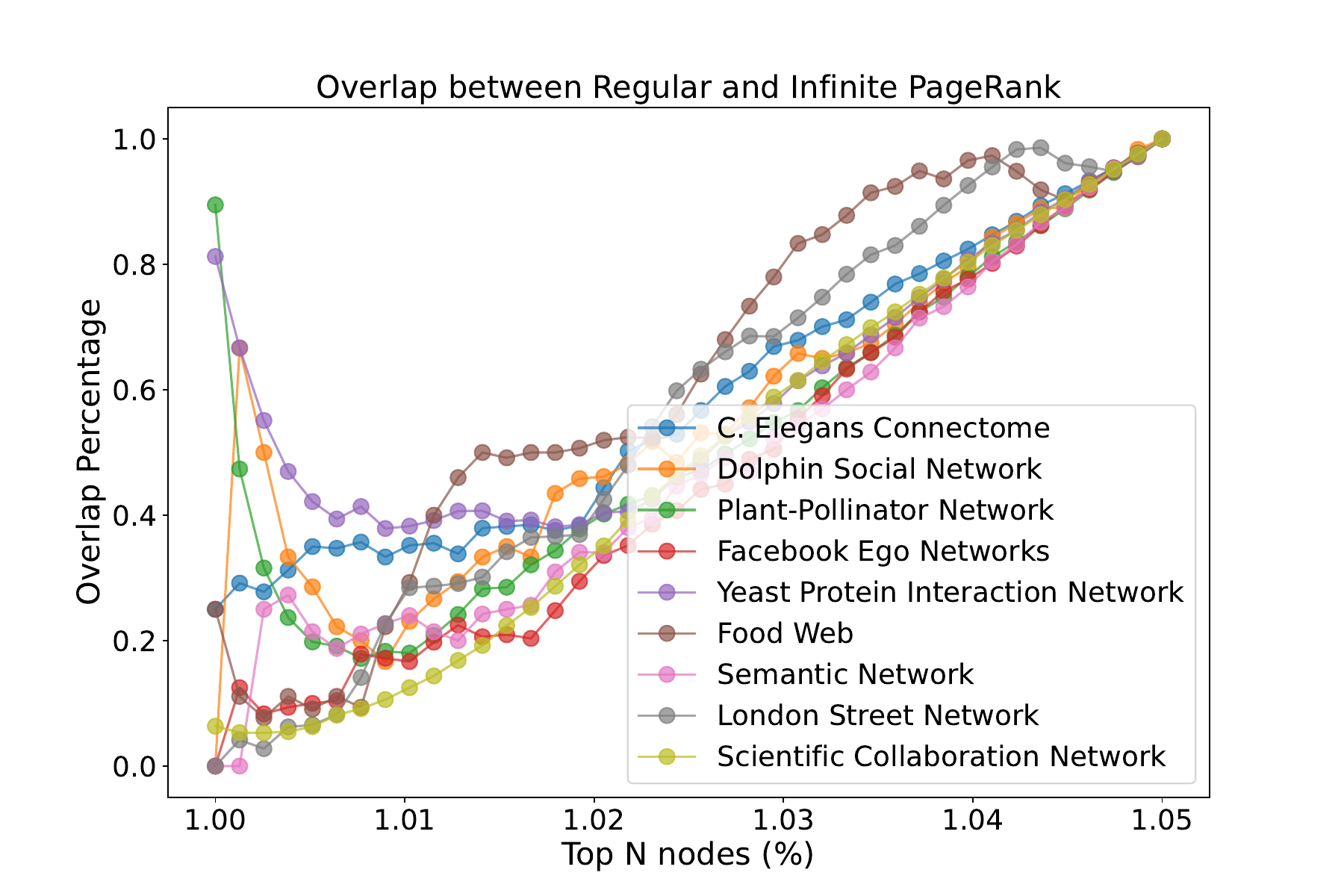}
  \caption{Percentage of overlap for top percentage of nodes in various \\real networks.}
\end{subfigure}}%
\usebox0\hfill\begin{minipage}[b][\ht0][s]{\wd0}
  \hspace*{-1cm}\begin{subfigure}{.5\linewidth}\centering
    \includegraphics[height=1.3in]{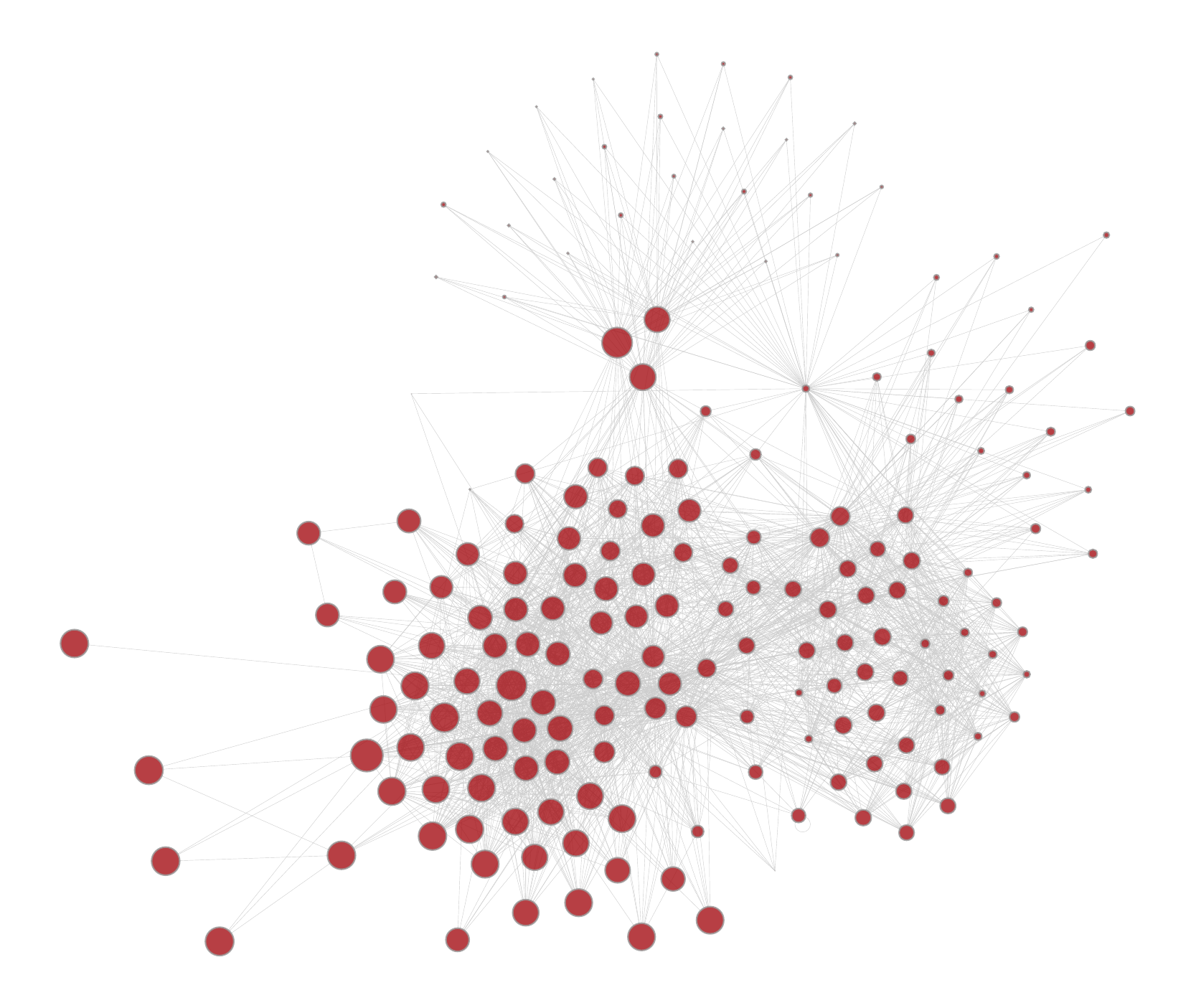}
    \caption{Foodweb where vertices are sized based on PageRank value.}
  \end{subfigure}\par
  \vfill
  \hspace*{-1cm}\begin{subfigure}[b]{.5\linewidth}\centering
    \includegraphics[height=1.3in]{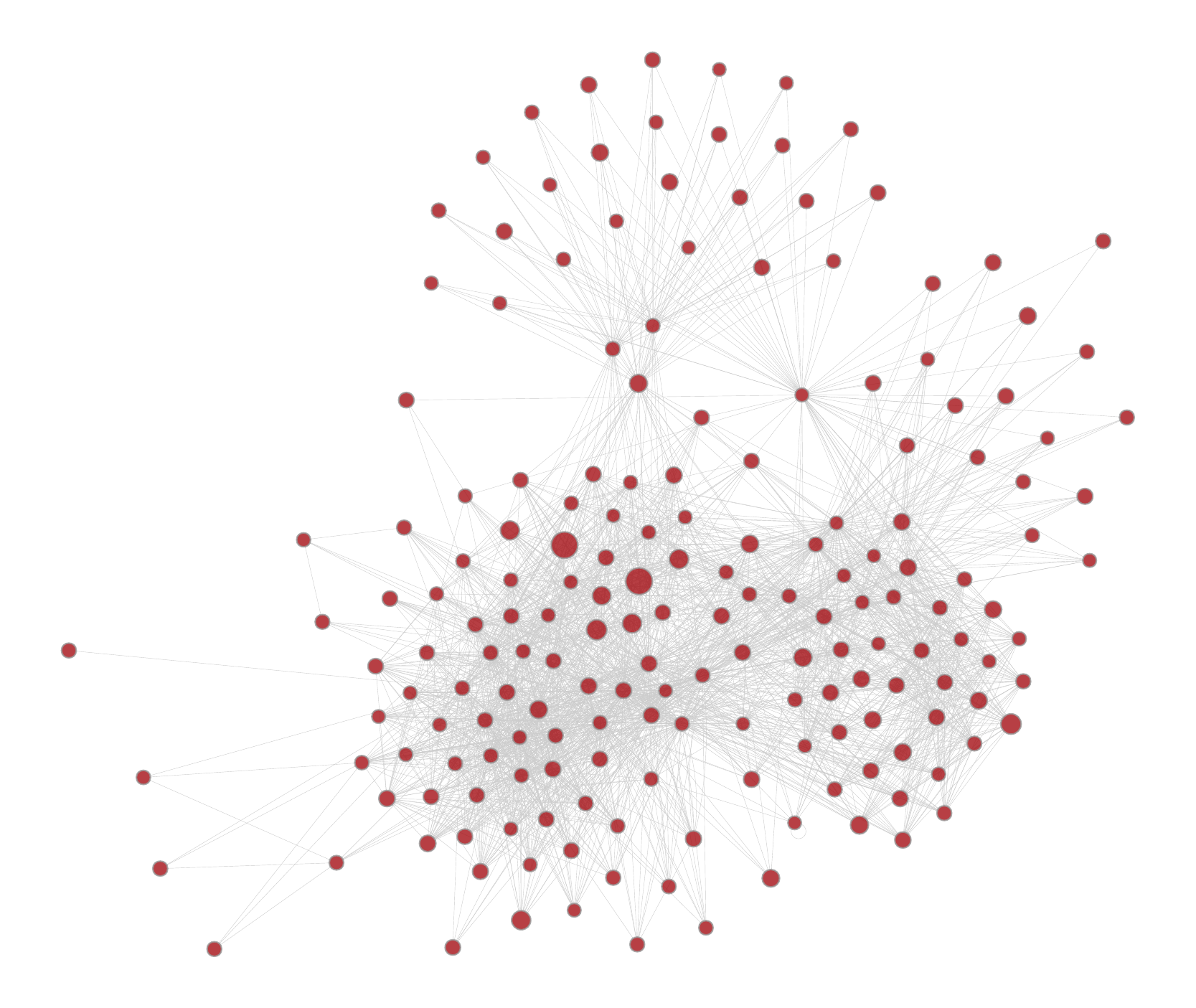}
    \caption{Foodweb where vertices are sized based on $\infty$-PageRank value.}
  \end{subfigure}
\end{minipage}
\caption{}
\label{fig:real-networks}
\end{figure}


\begin{remark}
    In a dense network, the computational complexity of calculating $\infty$-PageRank is $O(n^2)$ do to the matrix multiplication $AD^{-1}v$. In sparse networks this can be reduced to $O(n+m)$ where $m$ is the number of non-zero entries in $A$.
    In contrast, the computational complexity of calculate standard PageRank exactly $O(n^3)$. However standard computational methods such as the power method bring the complexity down to $O(k(n+m))$ on sparse networks, where $k$ is the number of iterations on the power method \cite{langville2011google}.
    Thus exact computation of $\infty$-PageRank is faster than approximation methods of standard PageRank.
    Addintionally, implementation of $\mu$-PageRank can be done by mapping the graph to its edge projection and applying standard packages. In the edge projection, a PageRank value must be calculated for every edge and thus in sparse networks will have complexity $O(k(m+n))$.
\end{remark}

\section{Clustering with $\infty$-PageRank}
\label{section:clustering}
To test the clustering capabilities of $\infty$-PageRank, we altered the algorithm PageRank-ClusteringA developed by Chung et al.~\cite{graham2010finding} to create an algorithm capable of $\infty$-PageRank clustering. In this process, we begin with a graph $G$ and its vertex set $V$. We then define an error tolerance value $tol$ that will determine when the algorithm has converged sufficiently for our needs, a jumping factor $\alpha\in(0,1)$ and a personalization vector $v_n$ for each node $n\in V$. In our experiments, we use $v_n=e_n$ where $e_n$ is the $n$th standard unit vector. Now for each node, we compute 
\begin{equation*}
    \rho_n = \rho_\infty(n,v_n) = (1+(1-\alpha))^{-1}v_n+(1-\alpha)(1+(1-\alpha))^{-1}AD^{-1}v_n
\end{equation*}
This is the $\infty$-PageRank vector for the node $n$ using personalization vector $v_n$. From here, we randomly choose $k$ centers and assemble their PageRank vectors \{$\rho_n$ for $n\in C_k$\} into a matrix $\Gamma$ where $\Gamma_i$ is a PageRank vector $\rho_n$. We are now ready to begin sorting the remaining $n-k$ nodes into our $k$ clusters. For each node $n$ we compute the \textit{PageRank Distance} of that node to each of the centers $\Gamma_i$ which is defined in \cite{graham2010finding} as 
\begin{equation*}
    \text{distance}_{n,i} = ||\rho_nD^{-1/2}-\Gamma_iD^{-1/2}||
\end{equation*}
From here we define a list of node labels $L$. The label of a node $n$ is determined by the value of $i$ which minimizes $\text{distance}_{n,i}$. We now redefine $\Gamma_i$ as the average of the PageRank vectors \{$\rho_n$|$L_n=i$\}, and we restart the iteration with our newly defined $\Gamma$. The iteration ends when the error between the previous version of $\Gamma$ and the most recent version of $\Gamma$ (which we will call $\Omega$) is within the tolerance we set at the beginning. At the end of the iteration, the labels assigned by $L$ are considered to be our final label estimates. The pseudocode for this process is given in Algorithm \ref{alg:cap}. 

\begin{algorithm}
\caption{$\infty$-PageRank Clustering Algorithm}\label{alg:cap}
\begin{algorithmic}
\State $V \gets$ Vertex set of Graph;\Comment{Initialize Variables}
\State $k \gets$ Number of Clusters;
\State $tol \gets$ Tolerance for Error;
\State $v_n \gets$ Personalization vector for node $n$;
\State $L \gets$ Collection of node labels;
\For{$n$ in $V$}
    \State $\rho_n \gets \rho_{\infty}(n,v_n)$ \Comment{Calculate $\infty$-PageRank with $v_n$}
\EndFor
\State $\Gamma \gets$ PageRank vectors of $k$ randomly selected nodes
\State error $\gets \infty$
\While{error $> tol$}
    \For{$n$ in $V$}
        \State distances$ \gets [0,0,...]$\Comment{Initialize list of distances of size $k$}
        \For{$i$ in range($k$)}
            \State $\text{distance}_i \gets\text{dist}(\rho_n,\Gamma_i)$ \Comment{Calculate PR-dist between node's PageRank vector and each cluster center}
        \EndFor
        \State $L_n \gets\text{argmin}(\text{distances})$\Comment{Label node based on which center it is closest to}
    \EndFor
    \State $C_i \gets$Collection of PageRank vectors of nodes with label $i$
    \State $\Omega_i \gets$Zero array of size $k$
    \For{$i$ in range($k$)}
        \State $\Omega_i \gets \text{ Average}(C_i)$\Comment{Calculate new cluster centers}
    \EndFor
    \State error $\gets \|\Gamma-\Omega\|$\Comment{Calculate difference in old and new cluster centers}
    \State $\Gamma\gets\Omega$ \Comment{Update centers}
\EndWhile
\end{algorithmic}
\end{algorithm}

\newpage

 We tested this algorithm on stochastic block matrices and the network containing 114 NCAA Division I American collegiate football teams, connected by the schools that had played against each other \cite{girvanfootball}. We will first discuss the outcome of the randomly generated stochastic block matrices. These were created by specifying a graph with $k$ clusters containing $n_i$ nodes in the $i$th cluster. The main metric we consider here is $c_{in}-c_{out}$, where $c_{in}$ is the probability of intra-cluster connections and $c_{out}$ is the probability of inter-cluster connections. So when $c_{in}-c_{out}=0$, it is equally likely for the node to be connected to a node within its cluster as it is to be connected to a node outside of its cluster. The lower $c_{in}-c_{out}$ values describe graphs without clear boundaries between clusters, making the clusters harder to detect. In Figure \ref{fig:clustering}, we show a stochastic block matrix (SBM) network with $c_{in}-c_{out}=0.8$ with the correct labels and with the clustering labels. Of the 90 nodes in this graph, approximately 89 were correctly clustered, giving an NMI of 0.98.
In total, we ran 450 clustering trials on randomly generated stochastic block matrices, and we charted the relationship between $c_{in}-c_{out}$ values and clustering algorithm accuracy measured via NMI (see Figure \ref{fig:clustering}). We see that there is almost a logarithmic relationship between the $c_{in}-c_{out}$ values and the accuracy of the clustering algorithm.  

\begin{figure}
\centering
\includegraphics[width=\textwidth]{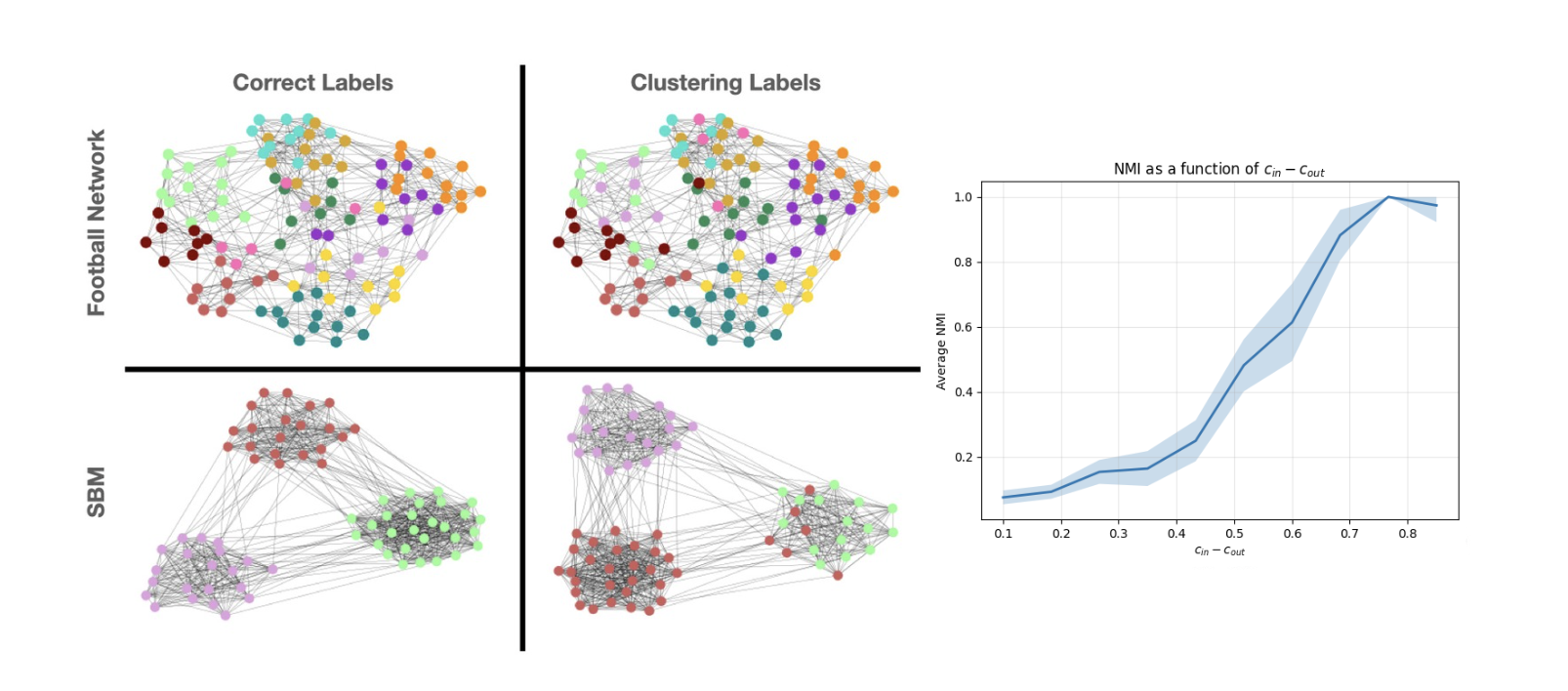}
\caption{The top row is the football network and the bottom row is a random stochastic block network.
The networks in the left column are the correctly labeled networks and the right are clustered networks.
The plot on the right shows the algorithm accuracy as measured by NMI as a function of the clustering strength in the network. The blue band shows the 95\% confidence interval.}
\label{fig:clustering}
\end{figure}

We now review the success of our algorithm on the Football network. This network of 114 nodes is a commonly used benchmark graph for new clustering algorithms. The goal is to be able to identify what conference each team corresponds to, thereby identifying 12 clusters. In Figure \ref{fig:clustering} the correct labels are shown, as well as the labels assigned by our algorithm. Our algorithm performed fairly well, with a normalized mutual information value of $0.8625$. For comparison, greedy modularity and the Louvain algorithm score $0.6977$ and $0.8505$ respectively \cite{barabasi2013network}.

\section{Conclusion}

We have investigated the relationship between PageRank and its variants as well as various properties of its variants.
Specifically, we have shown PageRank is equivalent to its variants in regular and bipartite biregular graphs.
Using these variants, we defined a new centrality measure, $\infty$-PageRank. 
Building off of the PageRank clustering work done by Chung  and Tsiatas \cite{graham2010finding}, we defined a new clustering algorithm based off of $\infty$-PageRank.

In addition to our findings, we have presented various conjectures which could prove vital for studying PageRank variants.
As future work, we hope to investigate these conjectures further to create more efficient comparison between standard and non-backtracking PageRank.\\

\noindent{\textbf{Data Availability Statement:}} The data and associated analysis tools from this paper are publicly available at \texttt{https://github.com/coryglover/backtracking\_pagerank}. The college football dataset used in the clustering analysis was taken from Mark Newman's public collection of network datasets at \texttt{http://www-personal.umich.edu/~mejn/netdata/}.\\

\noindent{\textbf{Declarations:}}  No authors have any financial or other conflicting interests.  All authors contributed to the conception, analysis, data generation, and writing of the manuscript.

\bibliographystyle{plain}
\bibliography{ref}

\pagebreak
\section*{Appendix}
\label{appendix}

In Section \ref{subsec:top-nodes} we measure the standard PageRank and $\infty$-PageRank distributions for six network models: (1) random graph model with parameters ($N=10^3$, $p=0.04$) \cite{erdds1959random}, (2) hypersoft configuration model with parameters ($N=10^3, \langle k\rangle=40, \gamma=2.1)$ \cite{voitalov2019scale}, random geometric graph with parameters ($N=10^3$, $r=0.04$) \cite{penrose2003random}, stochastic block model withh parameters ($N=10^3$, $C=10$, $p_{in}=0.04$, $p_{out}=0.001$) \cite{karrer2011stochastic}, Barabási-Albert preferential attachment model with parameters ($N=10^3$, $m=2$) \cite{barabasi1999emergence}, and a hyperbolic random graph model with parameters ($N=10^3$, $\beta=.5$, $\langle k\rangle=10$, $\gamma=2.01$) \cite{papadopoulos2012popularity}.
In Figure \ref{fig:models} we show the standard PageRank distribution (blue) and $\infty$-PageRank distribution (red) for each model.
We then examine the correlation plots between the standard and $\infty$-PageRank values in Figure \ref{fig:correlation-model} where we see strong correlation for heterogeneous networks specifically at high degree nodes.
In Figure \ref{fig:kendall-tau}a, we measure the Kendall-Tau correlation for the top $k$ percent of nodes in each model \cite{schaeffer1956concerning}.

In Section \ref{subsec:top-nodes} we measure $\infty$-PageRank on nine real world networks taken from the Netzschleuder database \cite{peixoto2020netzschleuder}.
Each network and its relevant statistics are found in Table \ref{table:networks}.
We then plot the correlation plots between standard and $\infty$-PageRank for each network in Figure \ref{fig:correlation-real-networks}.
We additionally measure the Kendall-Tau correlation for the top $k$ percent of nodes in each network in Figure \ref{fig:kendall-tau}b.

\begin{figure}
    \centering
    \begin{subfigure}{.45\linewidth}\centering
        \includegraphics[width=\linewidth]{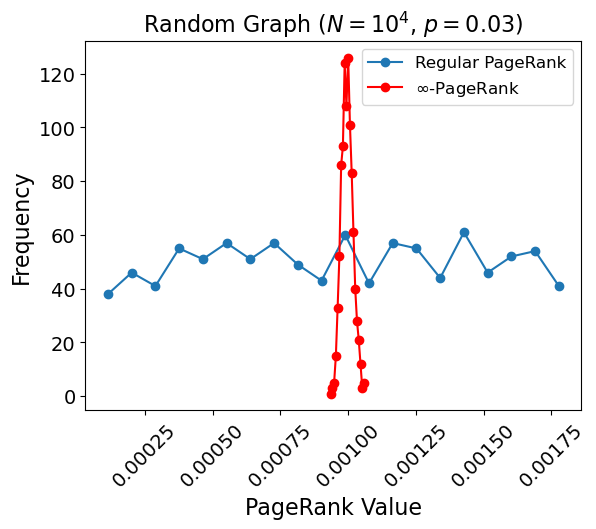}
        \caption{Random Graph}
    \end{subfigure}
    \begin{subfigure}{.45\linewidth}\centering
        \includegraphics[width=\linewidth]{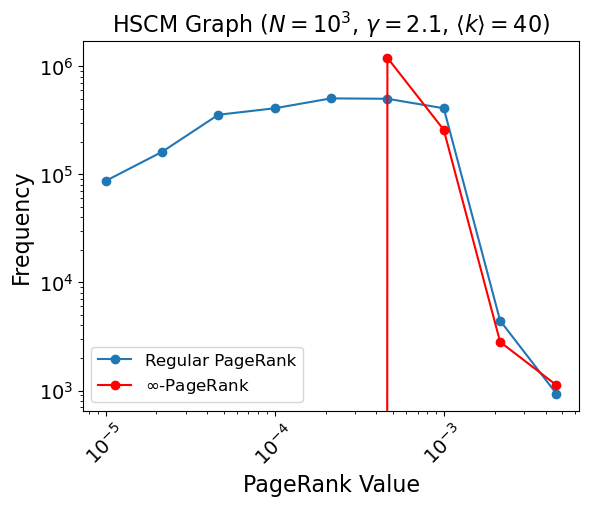}
        \caption{Hypersoft Configuration Model Graph}
    \end{subfigure}
    \begin{subfigure}{.45\linewidth}\centering
        \includegraphics[width=\linewidth]{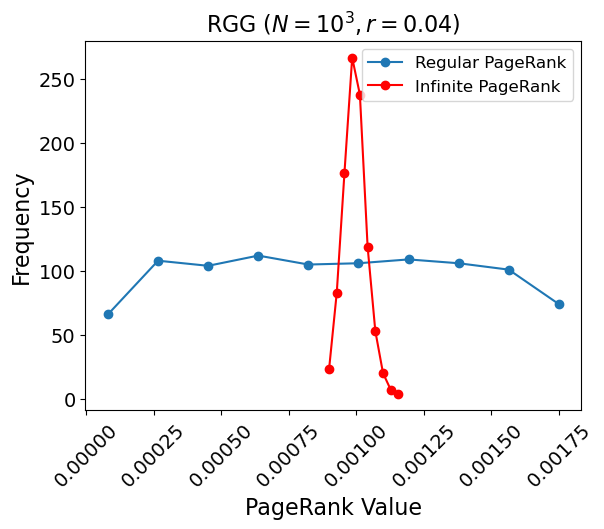}
        \caption{Random Geometric Graph}
    \end{subfigure}
    \begin{subfigure}{.45\linewidth}\centering
        \includegraphics[width=\linewidth]{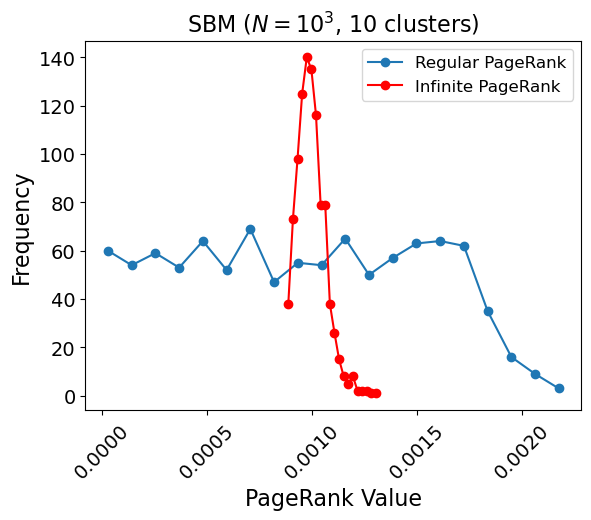}
        \caption{Stochastic Block Model Graph}
    \end{subfigure}
    \begin{subfigure}{.45\linewidth}\centering
        \includegraphics[width=\linewidth]{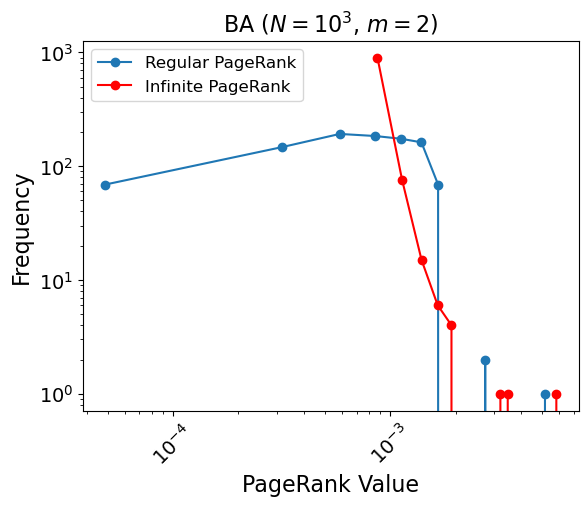}
        \caption{Barabási-Albert Graph}
    \end{subfigure}
    \begin{subfigure}{.45\linewidth}\centering
        \includegraphics[width=\linewidth]{hyperbolic_dist.png}
        \caption{Hyperbolic Random Graph}
    \end{subfigure}
    \caption{Standard versus $\infty$-PageRank Distributions}
    \label{fig:models}
\end{figure}

\begin{figure}
    \centering
    \begin{subfigure}{.45\linewidth}\centering
        \includegraphics[width=\linewidth]{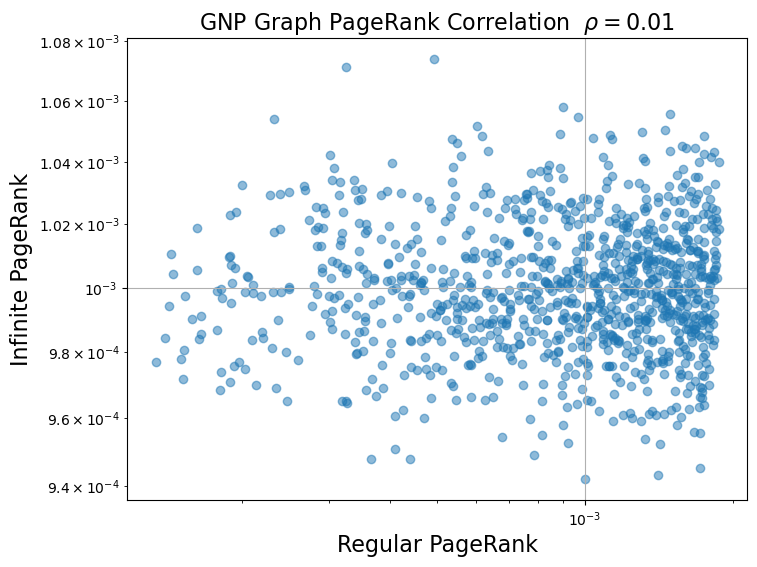}
        \caption{Random Graph}
    \end{subfigure}
    \begin{subfigure}{.45\linewidth}\centering
        \includegraphics[width=\linewidth]{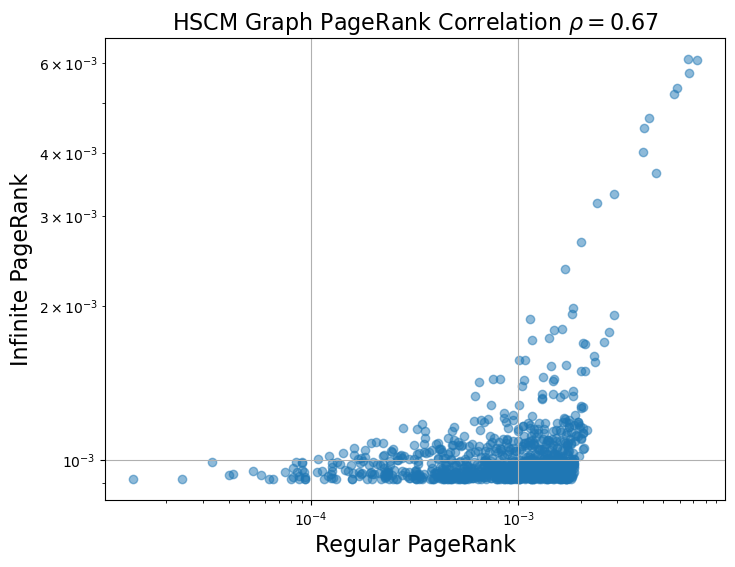}
        \caption{Hypersoft Configuration Model Graph}
    \end{subfigure}
    \begin{subfigure}{.45\linewidth}\centering
        \includegraphics[width=\linewidth]{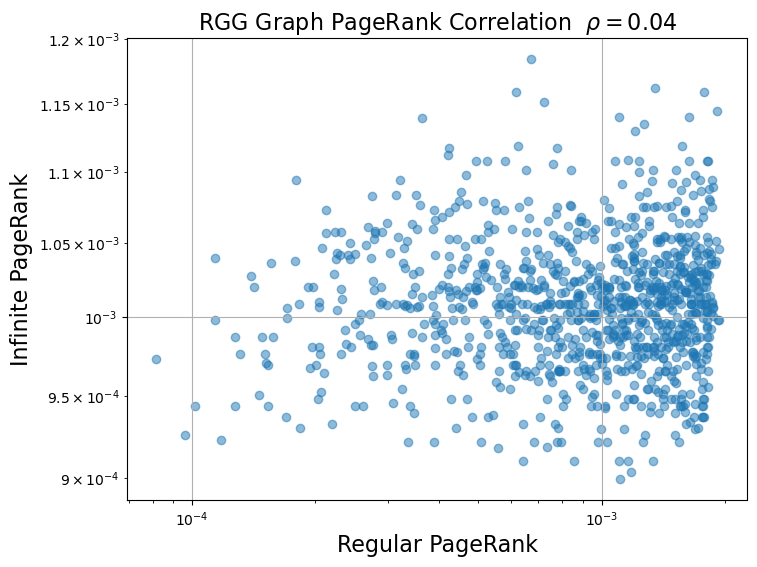}
        \caption{Random Geometric Graph}
    \end{subfigure}
    \begin{subfigure}{.45\linewidth}\centering
        \includegraphics[width=\linewidth]{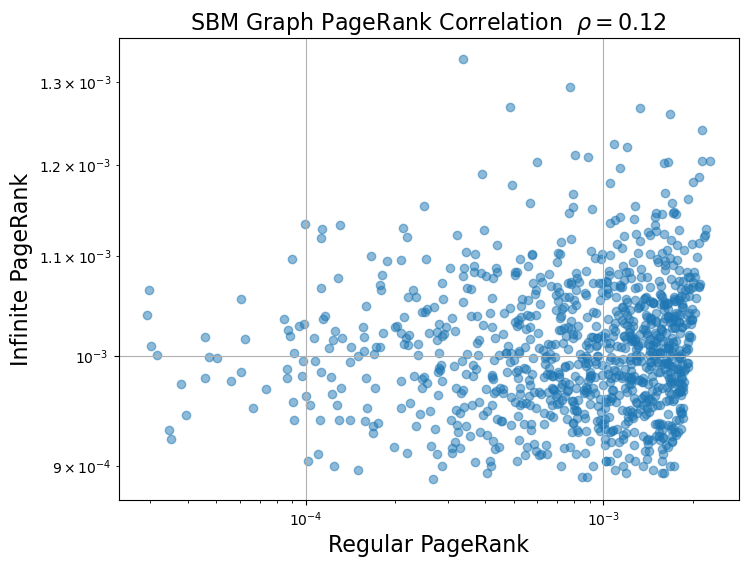}
        \caption{Stochastic Block Model Graph}
    \end{subfigure}
    \begin{subfigure}{.45\linewidth}\centering
        \includegraphics[width=\linewidth]{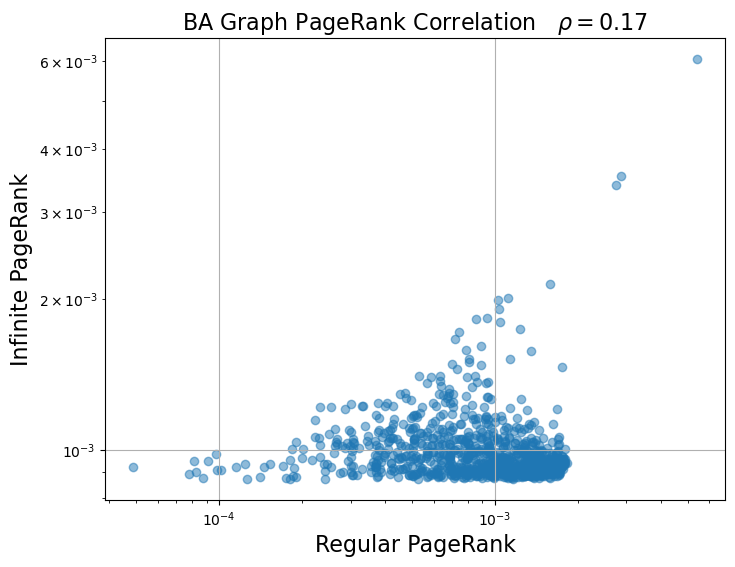}
        \caption{Barabási-Albert Graph}
    \end{subfigure}
    \begin{subfigure}{.45\linewidth}\centering
        \includegraphics[width=\linewidth]{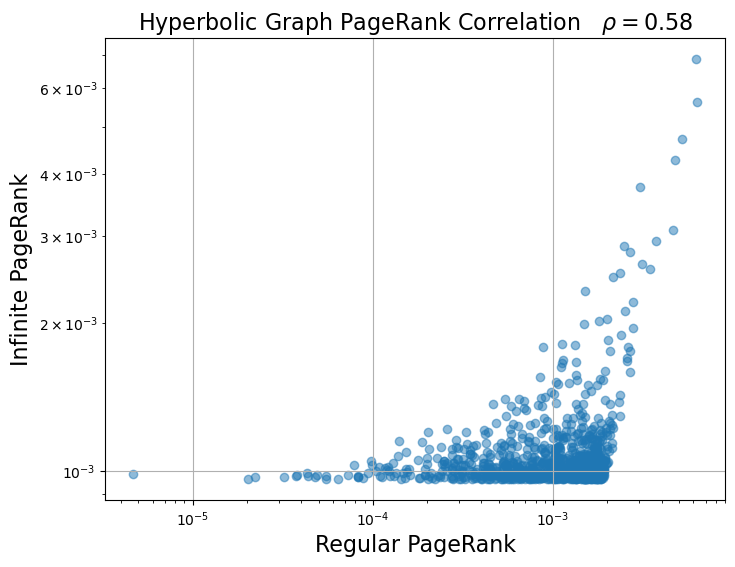}
        \caption{Hyperbolic Random Graph}
    \end{subfigure}
    \caption{Standard versus $\infty$-PageRank Correlation Plots}
    \label{fig:correlation-model}
\end{figure}

\begin{figure}
    \centering
    \begin{subfigure}{.45\linewidth}\centering
        \includegraphics[width=\linewidth]{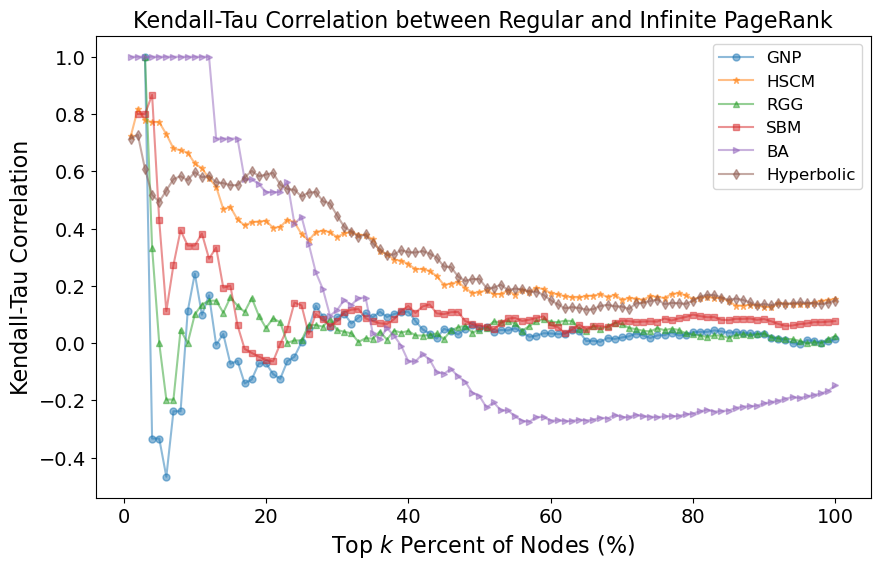}
        \caption{Graph Models}
    \end{subfigure}
    \begin{subfigure}{.45\linewidth}\centering
        \includegraphics[width=\linewidth]{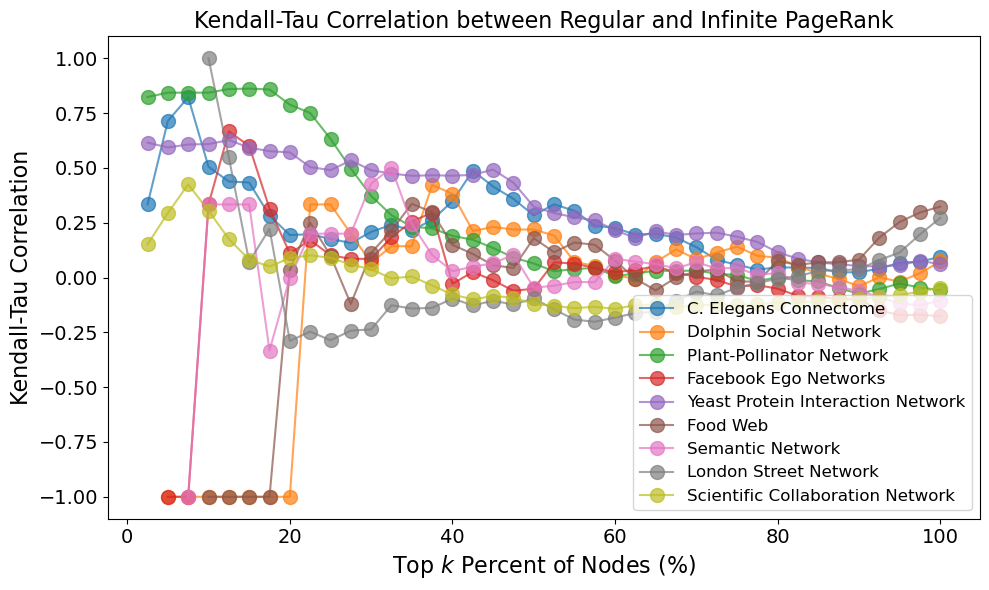}
        \caption{Real Networks}
    \end{subfigure}
    \caption{Kendall-Tau Correlation for top $k$\% of vertices}
    \label{fig:kendall-tau}
\end{figure}

\begin{table}[b]
    \centering
    \begin{tabular}{c|c|c|c|c}
        Network & $N$ & $L$ & $\langle k\rangle$ & $\langle k^2\rangle$ \\
        \hline
        C.~Elegans Connectome & 585 & 1755 & 6 & 81.12\\
        Dolphin Social Network & 62 & 159 & 5.13 & 34.9\\
        Plant Pollinator Network & 772 & 1206 & 3.12 & 95.36 \\ 
        Facebook Ego Network & 333 & 2519 & 15.13 & 469.53 \\
        Yeast Transcription Network & 916 & 1094 & 2.39 & 3.32\\
        Little Rock Food Web & 183 & 2494 & 27.26 & 292.67\\
        Word Adjacencies of David Copperfield & 112 & 425 & 7.59 & 104.54 \\
        London Street Network & 488 & 730 & 2.99 & 9.81 \\
        ArXiV Collaboration Network & 16726 & 47594 & 5.69 & 73.57
    \end{tabular}
    \caption{Network statistics for real networks analyzed}
    \label{table:networks}
\end{table}

\begin{figure}
    \centering
    \begin{subfigure}{.32\linewidth}\centering
        \includegraphics[width=\linewidth]{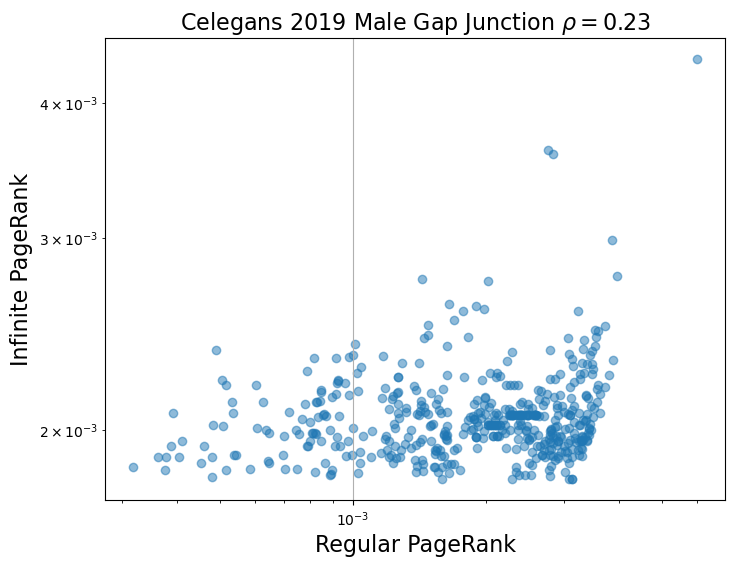}
        \caption{C.~Elegans Connectome}
    \end{subfigure}
    \begin{subfigure}{.32\linewidth}\centering
        \includegraphics[width=\linewidth]{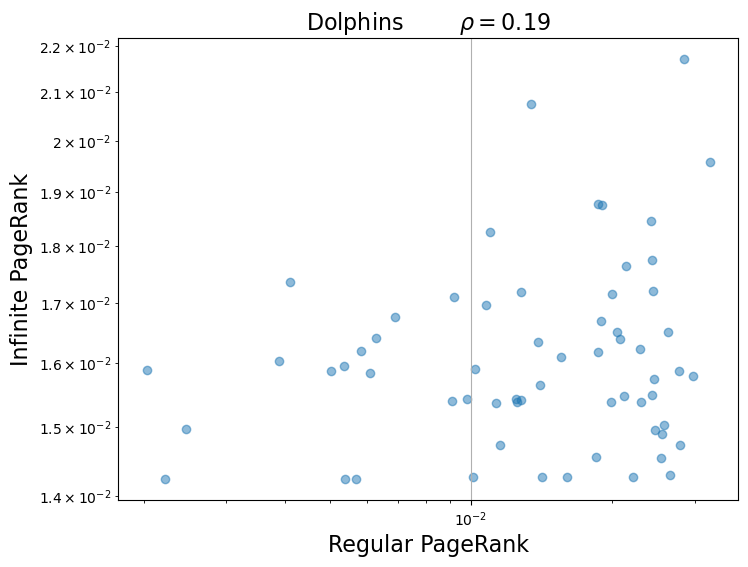}
        \caption{Dolphin Social Network}
    \end{subfigure}
    \begin{subfigure}{.32\linewidth}\centering
        \includegraphics[width=\linewidth]{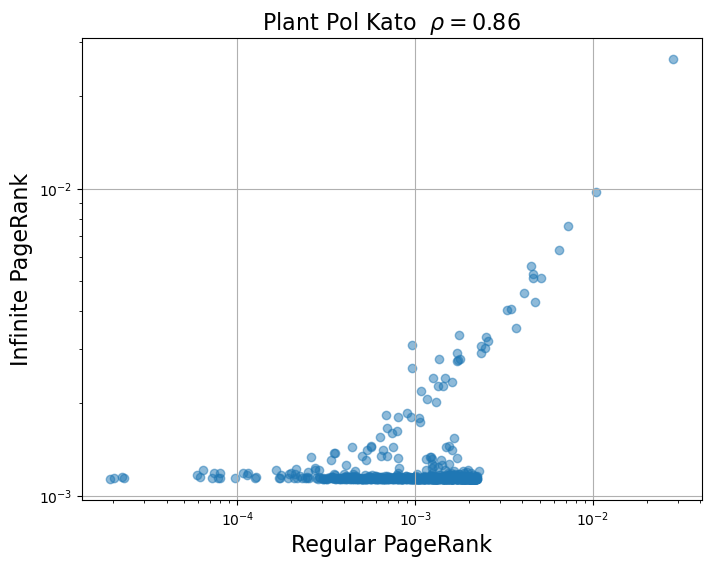}
        \caption{Plant-Pollinator Network}
    \end{subfigure}
    \begin{subfigure}{.32\linewidth}\centering
        \includegraphics[width=\linewidth]{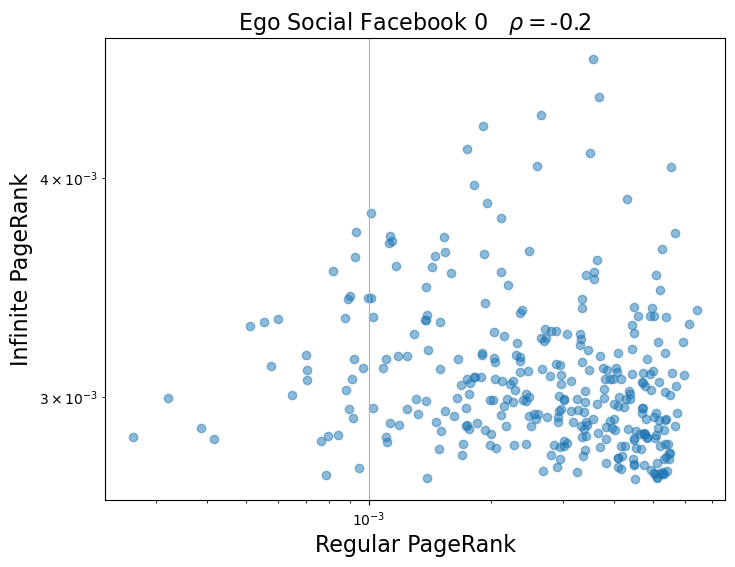}
        \caption{Facebook Ego Network}
    \end{subfigure}
    \begin{subfigure}{.32\linewidth}\centering
        \includegraphics[width=\linewidth]{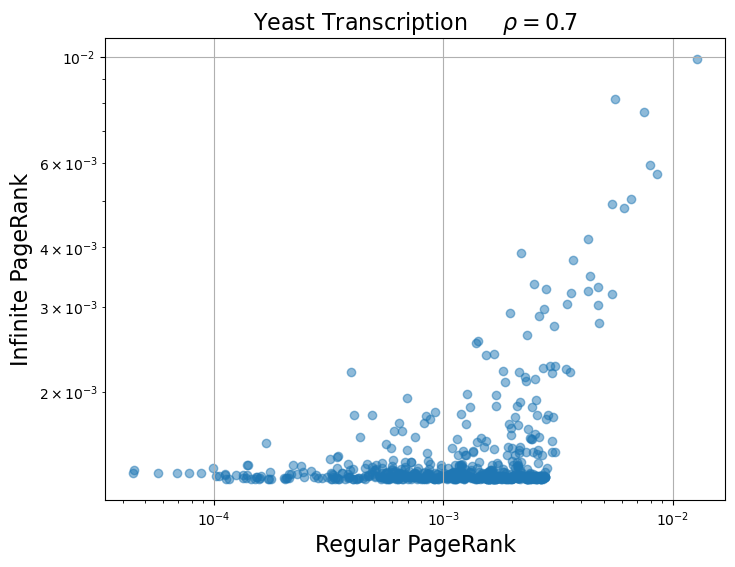}
        \caption{Yeast Transcription Network}
    \end{subfigure}
    \begin{subfigure}{.32\linewidth}\centering
        \includegraphics[width=\linewidth]{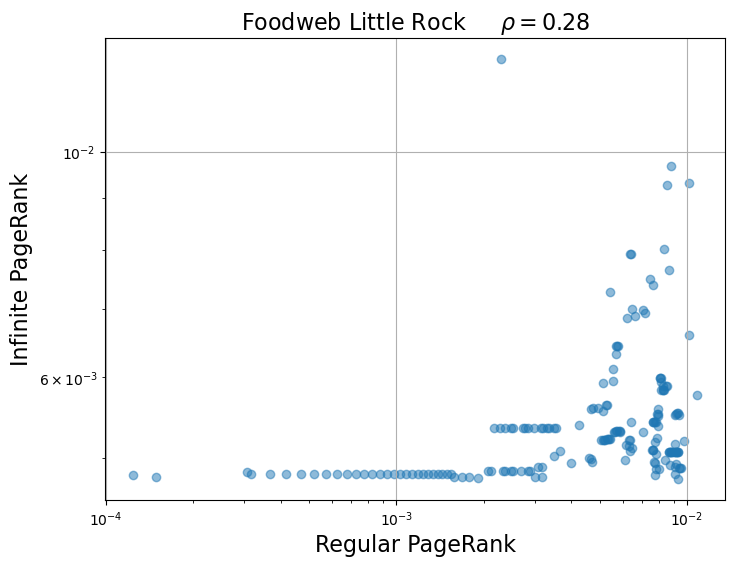}
        \caption{Little Rock Food Web}
    \end{subfigure}
    \begin{subfigure}{.32\linewidth}\centering
        \includegraphics[width=\linewidth]{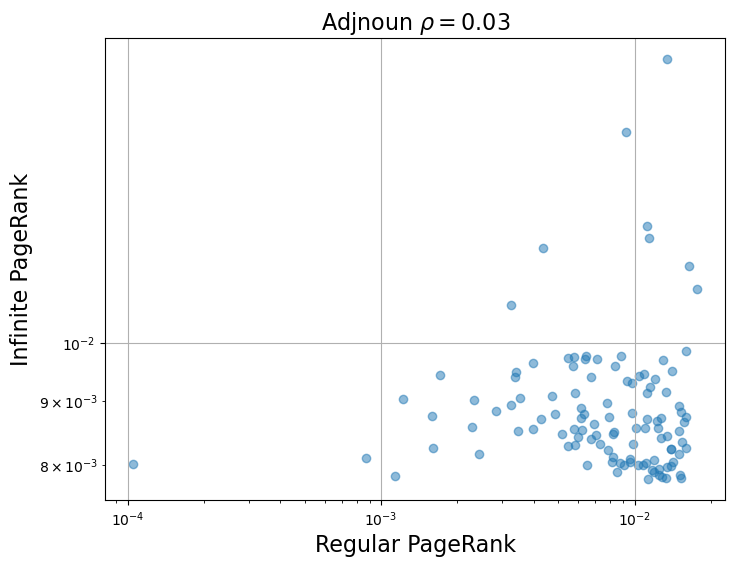}
        \caption{Semantic Network}
    \end{subfigure}
    \begin{subfigure}{.32\linewidth}\centering
        \includegraphics[width=\linewidth]{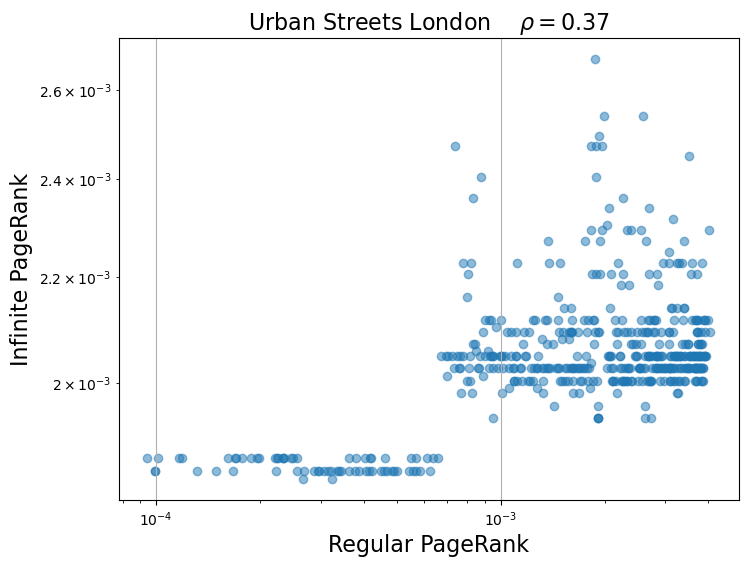}
        \caption{London Street Network}
    \end{subfigure}
    \begin{subfigure}{.32\linewidth}\centering
        \includegraphics[width=\linewidth]{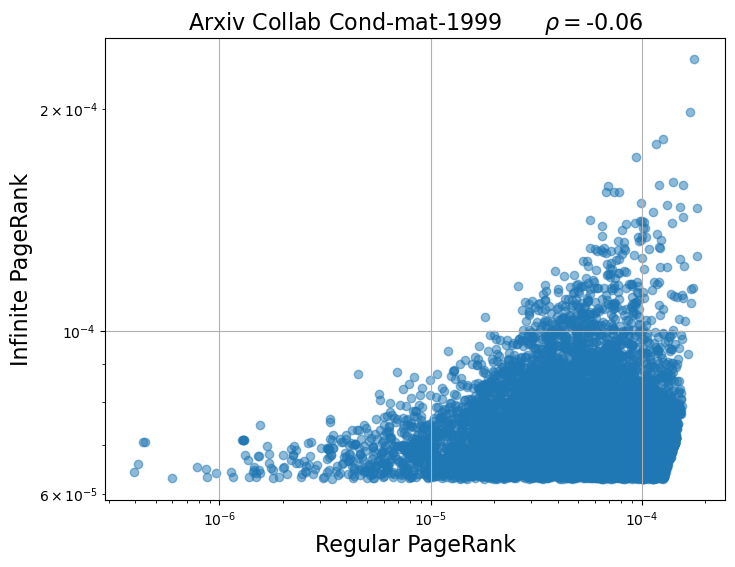}
        \caption{Scientific Collaboration Network}
    \end{subfigure}
    \caption{Standard versus $\infty$-PageRank Correlations}
    \label{fig:correlation-real-networks}
\end{figure}

\end{document}